\documentclass[12pt]{article}

\usepackage{newtxtext,newtxmath}

\usepackage{graphicx}

\usepackage[letterpaper,margin=1in]{geometry}

\renewenvironment{abstract}
	{\quotation}
	{\endquotation}

\date{}

\makeatletter
\renewcommand{\fnum@figure}{\textbf{Figure \thefigure}}
\renewcommand{\fnum@table}{\textbf{Table \thetable}}
\makeatother
\usepackage{subcaption}
\usepackage{scicite}
\usepackage{url}

\newenvironment{proof}[1][Proof]{\noindent\textbf{#1.} }{\ \rule{0.5em}{0.5em}}
\newtheorem{theorem}{Theorem}%  meant for continuous numbers
\newtheorem{lemma}{Lemma}
\newtheorem{remark}{Remark}%
\newtheorem{assumption}{Assumption}

\def\scititle{
	Enabling Green Wireless Communications with Neuromorphic Continual Learning
}

\title{\bfseries \boldmath \scititle}

\author{
	Yanzhen Liu$^{1}$,
	Zhijin Qin$^{2\ast}$,
	Yongxu Zhu$^{3}$, and
	Geoffrey Ye Li$^{1}$\and
	\small$^{1}$Department of Electrical and Electronic Engineering, Imperial College London, London \& SW7 2AZ, UK.\and
	\small$^{2}$Department of Electronic Engineering, Tsinghua University, Beijing \& 100084, China.\and
	\small$^{3}$National Mobile Communications Research Laboratory, Southeast University, Nanjing \& 210096, China.\and
	\small$^\ast$Corresponding author. Email: qinzhijin@tsinghua.edu.cn
}

\begin{document} 

% Insert the title and author list
\maketitle

\begin{abstract} \bfseries \boldmath	
The pursuit of carbon-neutral wireless networks is increasingly constrained by the escalating energy demands of deep learning-based signal processing. Here, we introduce SpikACom (\textbf{Spik}ing \textbf{A}daptive \textbf{Com}munications), a neuromorphic computing framework that synergizes brain-inspired spiking neural networks (SNNs) with wireless signal processing to deliver sustainable intelligence. SpikACom advances the paradigm shift from energy-intensive, continuous-valued processing to event-driven sparse computation. Moreover, it supports continual learning in dynamic wireless environments via a dual-scale mechanism that integrates channel distribution-aware context modulation with a synaptic consolidation rule using SNN-specific statistics, mitigating catastrophic forgetting. Evaluations across critical wireless communication tasks, including semantic communication, multiple-input multiple-output (MIMO) beamforming, and channel estimation demonstrate that SpikACom matches full-precision deep learning baselines while achieving an order-of-magnitude improvement in computational energy efficiency. Our results position SNNs as a promising pathway toward green wireless intelligence, providing evidence that neuromorphic computing can empower the sustainability of modern digital systems.
\end{abstract}

\noindent
\section*{Introduction}\label{sec1}
Wireless networks are evolving to deliver not only connectivity but also ubiquitous intelligence, with deep learning (DL) rapidly moving into the physical layer and end-to-end communication stacks\cite{dlforphy}. Over the past decade, DL has been successfully applied to a wide range of wireless communication tasks, such as channel estimation \cite{OFDM_CE1}, channel state information (CSI) feedback \cite{dlcsifeedback}, symbol detection \cite{symboldetection}, waveform design \cite{dlwaveform}, resource allocation \cite{dlresourceallocation}, and cross-layer semantic communications \cite{semantic1,multimodalcom,foundation1,holographicsemantic}. This trend is also reflected in standardization. For example, the 3rd Generation Partnership Project (3GPP) has launched projects exploring artificial intelligence (AI) for the fifth generation (5G) air interface~\cite{3GPP}, identifying pilot use cases such as CSI feedback, beam management, and positioning. However, although learning-based methods can improve performance under complex channels and hardware impairments, they rely on dense floating-point computation, raising concerns that computational energy will become a key limiting factor for sustainable wireless systems. This concern is particularly acute because the computational cost of signal processing has long been a major burden, well before deep learning was introduced into wireless communications \cite{MIMO_energy}. With radio access networks accounting for more than 70\% of total mobile network power consumption \cite{RAN_energy}, there is little headroom for computationally heavy processing at the wireless edge, especially under stringent carbon-neutrality targets that require the information and communication technology (ICT) sector to reduce CO$_2$ emissions by 91\% by 2050 relative to 2010 levels~\cite{ICT2020industry,ICTclimate}. This tension motivates the search for alternative computing paradigms to deliver wireless intelligence at substantially lower energy cost.

Spiking neural networks (SNNs) have emerged as a compelling solution for pursuing sustainable AI-native wireless communications, marking a fundamental paradigm shift toward event-driven communication systems (Fig.~\ref{fig:paradiagm_shift}). Unlike traditional artificial neural networks (ANNs), SNNs emulate the biological brain’s information processing~\cite{snn_nature} by employing biophysical neuron models that communicate via discrete events known as spikes (Fig.~\ref{fig:SNN_model}). This event-driven nature allows SNNs to trigger computation primarily when informative events arrive, thereby avoiding redundant processing. In addition, the binary nature of spikes converts the power-hungry multiply-accumulate (MAC) operations dominant in ANNs into simple accumulations (ACs), significantly reducing energy consumption. Beyond computation, the spike representation also offers flexibility in signal transmission: spikes can be mapped to pulse-based physical-layer waveforms (e.g. impulse radios) \cite{uwb} or encapsulated into digital bit sequences for seamless integration with existing digital communication pipelines \cite{snnprototype}. Collectively, these elements point to an end-to-end neuromorphic wireless communication system, enabling networks to process and exchange high-level semantic meaning with brain-inspired efficiency.

Driven by this vision, SNNs have attracted increasing attention in the wireless community. Recent studies have successfully applied SNNs to domains ranging from distributed neuromorphic edge computing~\cite{snndistributed1}, federated learning~\cite{snnfederated,snnleadfed}, and physical-layer signal processing~\cite{snnisac,snn_VDES,snn_satellite,snnofdmreceiver} to emerging paradigms like task-oriented semantic communication~\cite{snnsemantic3,snnsemantic1,snnsemantic2,multilevelsnn}. Despite these advancements, the application of SNNs in communication systems remains at an early stage. Early works pioneered neuromorphic joint source-channel coding~\cite{snnsemantic3,snnsemantic1}, providing important feasibility demonstrations, but they primarily focused on constrained tasks such as classifying two digits in the MNIST-DVS dataset \cite{mnist-dvs}. Subsequently, methods based on ANN-to-SNN conversion~\cite{ANN2SNN} were introduced, successfully applying SNNs to radio frequency fingerprint identification. However, ANN-to-SNN conversion imposes restrictive constraints on neuron models and often requires long simulation time, which limits the full exploitation of the energy efficiency of SNNs. The potential of SNN-based signal processing was further unlocked by the introduction of surrogate gradient-based training~\cite{snnsurrogate}. This approach has demonstrated empirical success on more realistic tasks, such as satellite symbol detection~\cite{snn_satellite} and orthogonal frequency-division multiplexing (OFDM) receivers~\cite{snnofdmreceiver}. Nevertheless, these works largely reuse well-established computer vision architectures, with limited incorporation of communication domain priors. Notably, while spike-based computation can be highly energy-efficient, SNNs’ representational capacity for complex tasks is often lower than that of full-precision ANNs. The capability of SNNs to handle sophisticated, domain-specific wireless communication problems remains an open challenge. Beyond these design limitations, relatively few studies explicitly address the inherent challenges of deploying SNNs in wireless environments. In practice, wireless channels shift on short time scales, requiring continual model updates. However, such continual adaptation renders the model susceptible to catastrophic forgetting, where new updates overwrite the representations of previously learned channel distributions. Developing stable adaptation schemes that mitigate this forgetting, while maintaining bounded computational and memory overhead, remains largely unexplored in current SNN-based solutions.

Motivated by these limitations, we introduce SpikACom (\textbf{Spik}ing \textbf{A}daptive \textbf{Com}munications), a general framework for integrating SNNs into wireless communication systems (Fig.~\ref{fig:dynamic_wireless_env}). The framework supports continual learning in dynamic wireless environments via a dual-scale adaptation mechanism that combines (i) a hypernet-based context modulator that controls the backbone SNN through channel distribution-aware gating signals and (ii) a synaptic-level consolidation rule that regularizes updates using SNN-specific statistics (Fig.~\ref{fig:spikacom_framework}), effectively mitigating catastrophic forgetting with relatively low computational and memory overhead. We evaluate SpikACom on three representative wireless communication tasks: neuromorphic semantic communication, multi-user multiple-input multiple-output (MIMO) beamforming, and channel estimation in OFDM systems. These tasks are aligned with key AI use cases in 3GPP standardization \cite{3GPP}, particularly CSI-related enhancement and beam management, while also covering emerging cross-layer paradigms. We demonstrate that SNNs can move beyond proof-of-concept tasks to tackle sophisticated wireless signal processing problems with inherent physical structures (e.g., MIMO beamforming and OFDM channel estimation), showing that incorporating domain knowledge unlocks the capability of sparse spike-based computing to capture complex signal dependencies without compromising order-of-magnitude energy savings. Overall, our work highlights the potential of neuromorphic computing to enable more sustainable, adaptive, and intelligent wireless communication systems, motivating energy-efficient SNN-native architectures for the sixth generation (6G) and beyond.

\begin{figure}
	\centering
	
	\begin{subfigure}[b]{0.95\textwidth}
		\centering
		\includegraphics[width=\textwidth]{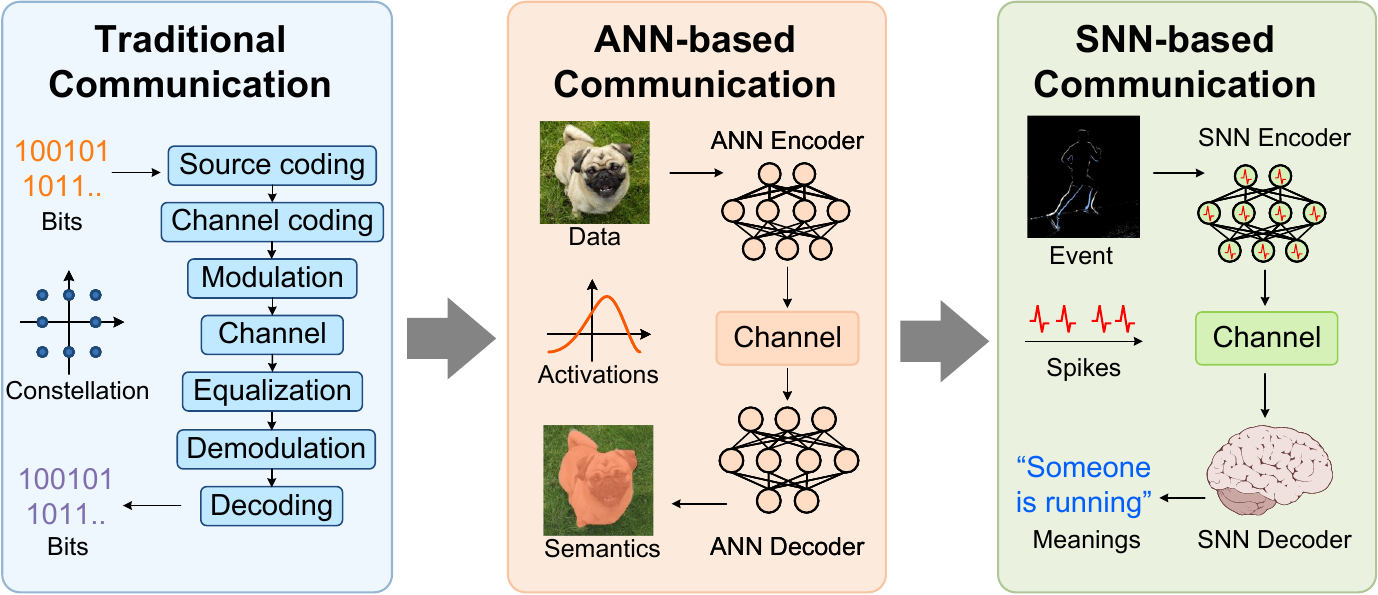}
		\caption{}
		\label{fig:paradiagm_shift}
		\vspace{-0.7em}
	\end{subfigure}

	\begin{subfigure}[b]{0.40\textwidth}
		\centering
		\raisebox{2em}{
			\includegraphics[width=0.95\textwidth]{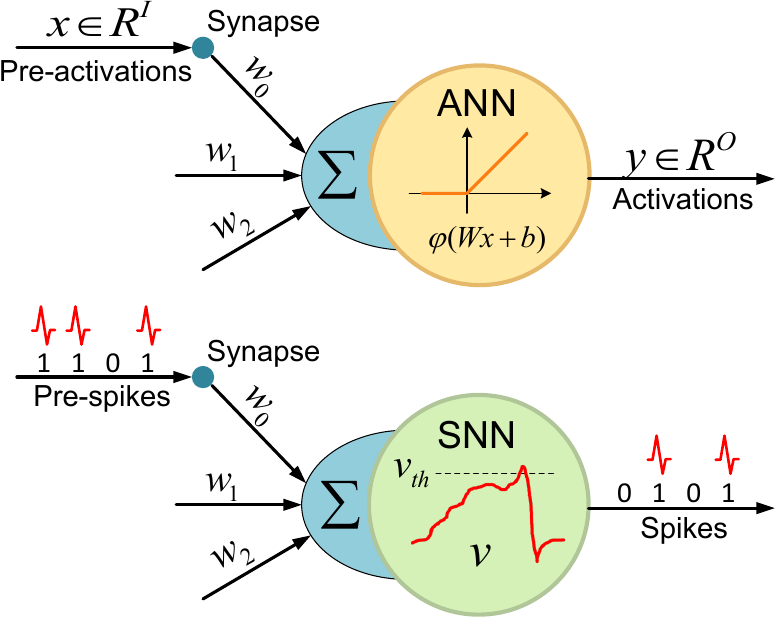}
		}
		\caption{}
		\label{fig:SNN_model}
	\end{subfigure}	
	\hfill
	\begin{subfigure}[b]{0.56\textwidth}
		\centering
		\includegraphics[width=0.95\textwidth]{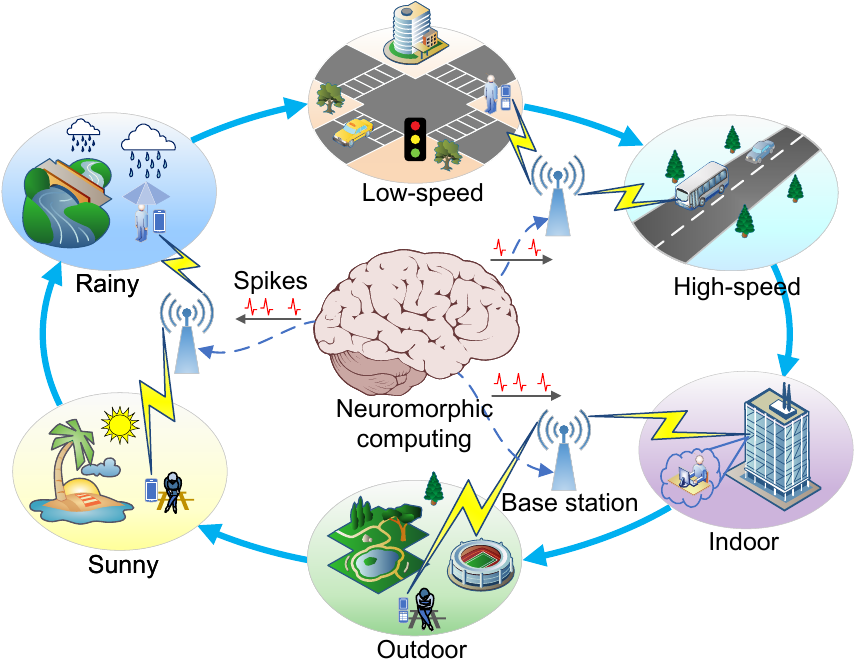}
		\caption{}
		\label{fig:dynamic_wireless_env}
		\vspace{-0.3em}
	\end{subfigure}

	\caption{\linespread{1.0}\selectfont \textbf{Illustration of neuromorphic computing empowered wireless systems}. 
		\textbf{(a)} Evolution of communication paradigms. \textbf{Left:} Traditional communication focuses on bit-level reliability and is built upon separate model-driven signal processing blocks, where information is mapped to constellation symbols for transmission. \textbf{Middle:} ANN-based communication employs dense neural networks to replace or augment conventional processing blocks, enabling information processing at higher levels of abstraction. The transmitted representations are intermediate continuous-valued activations. \textbf{Right:} SNN-based (neuromorphic) communication shifts toward event-driven processing, where information is encoded and conveyed as sparse binary spikes, enabling energy-efficient operation and facilitating semantics-oriented, cognition-inspired communication.
		\textbf{(b)} Comparison between an artificial neuron and a spiking neuron. The spiking neuron communicates via discrete binary spikes and inherently operates in the spatio-temporal domain, whereas the artificial neuron computes with continuous-valued activations.
		\textbf{(c)} Neuromorphic computing enabled wireless communication system illustrating low-power operation and continual adaptation in dynamic wireless environments affected by factors such as mobility, weather, and environmental scattering.}
	\label{fig:main_figure}
\end{figure}

\section*{Results}\label{sec2}
\begin{figure}
	\centering
	
	\includegraphics[width=0.95\textwidth]{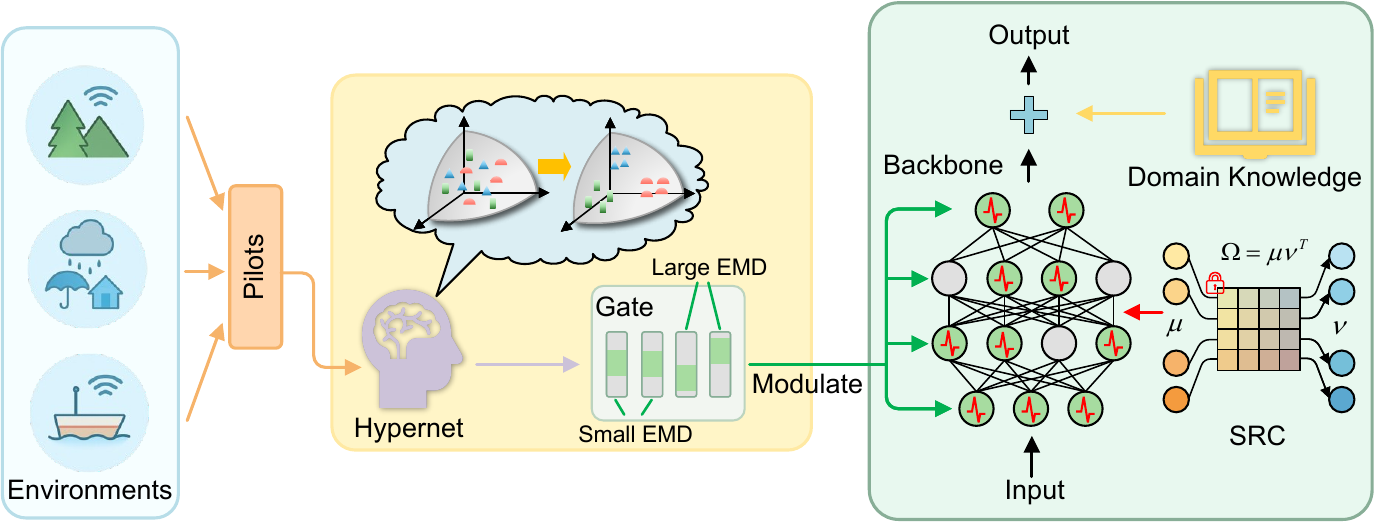}
	
	\caption{\linespread{1.0}\selectfont \textbf{Overview of the SpikACom framework}. \textbf{Left}: Pilots collected in different wireless conditions are used to identify the environment. \textbf{Middle}: A hypernet-based context modulator senses channel distribution shifts (quantified via earth mover's distance (EMD)) based on the pilots and generates binary gates to modulate the backbone SNN. \textbf{Right}: The backbone SNN processes the task inputs and supports task-dependent incorporation of domain knowledge to enhance performance, while spiking rate consolidation (SRC) regularizes synaptic weights based on neuronal activity to mitigate catastrophic forgetting.}
	
	\label{fig:spikacom_framework}
\end{figure}

\subsection*{Continual Learning in Wireless Environments}
In real-world deployments, signals propagate through complex and continually changing environments: user mobility, blockage/obstacles, and even weather can induce time-varying channel conditions. As a result, models trained offline can degrade rapidly when exposed to distribution shifts in the channel. Therefore, learning-based communication systems must detect such changes and adapt continually. In practice, however, many models are deployed on edge devices with limited memory and cannot store large amounts of historical data. 
They must update their parameters using only the current observations, a setting commonly referred to as continual learning~\cite{continuous_learning}. 
A central difficulty in continual learning is catastrophic forgetting, where adapting to new conditions causes the model to overwrite previously acquired knowledge. 
This phenomenon results in a fundamental trade-off between plasticity (learning new information) and stability (retaining past knowledge). 
Severe forgetting forces the model to relearn previously encountered scenarios, leading to significant computational waste and preventing sustained improvement over time.

We propose SpikACom to enable continual learning of SNNs in wireless communication systems. The framework is structured as a dual-scale adaptation mechanism (Fig.~\ref{fig:spikacom_framework}). At the context level, SpikACom introduces a lightweight hypernet-based context modulator that conditions the backbone SNN on the current channel distribution. This design is motivated by the fact that wireless channels, while varying across environments, are governed by common propagation principles, which induce structured correlations among channel distributions. By learning a compact representation of the environment context, the modulator separates learning across different regimes, reduces destructive interference, and encourages parameter reuse when environments are similar. At the synaptic level, SpikACom employs spiking rate consolidation (SRC), a neuromorphic regularization strategy inspired by Hebbian principles. SRC uses neuron spiking rate as an SNN-specific proxy for parameter importance, selectively stabilizing critical synaptic pathways during adaptation. Because it relies on readily available spiking statistics rather than gradient-based information, SRC incurs substantially lower computation and memory overhead than conventional regularization-based continual learning methods such as elastic weight consolidation (EWC)~\cite{EWC}. Together, this hierarchical design enables SpikACom to adapt rapidly with remarkably low sample requirements and minimal computational overhead, providing a lightweight solution tailored to the resource-constrained nature of wireless communication systems.

In the following subsections, we evaluate SpikACom on three representative wireless communication tasks: neuromorphic semantic communication, multi-user MIMO beamforming, and channel estimation. 

\subsection*{Neuromorphic Semantic Communication}\label{subsubsec2}

\begin{figure}
	\centering

	\begin{subfigure}[b]{0.96\textwidth}
		\centering
		\includegraphics[width=\textwidth]{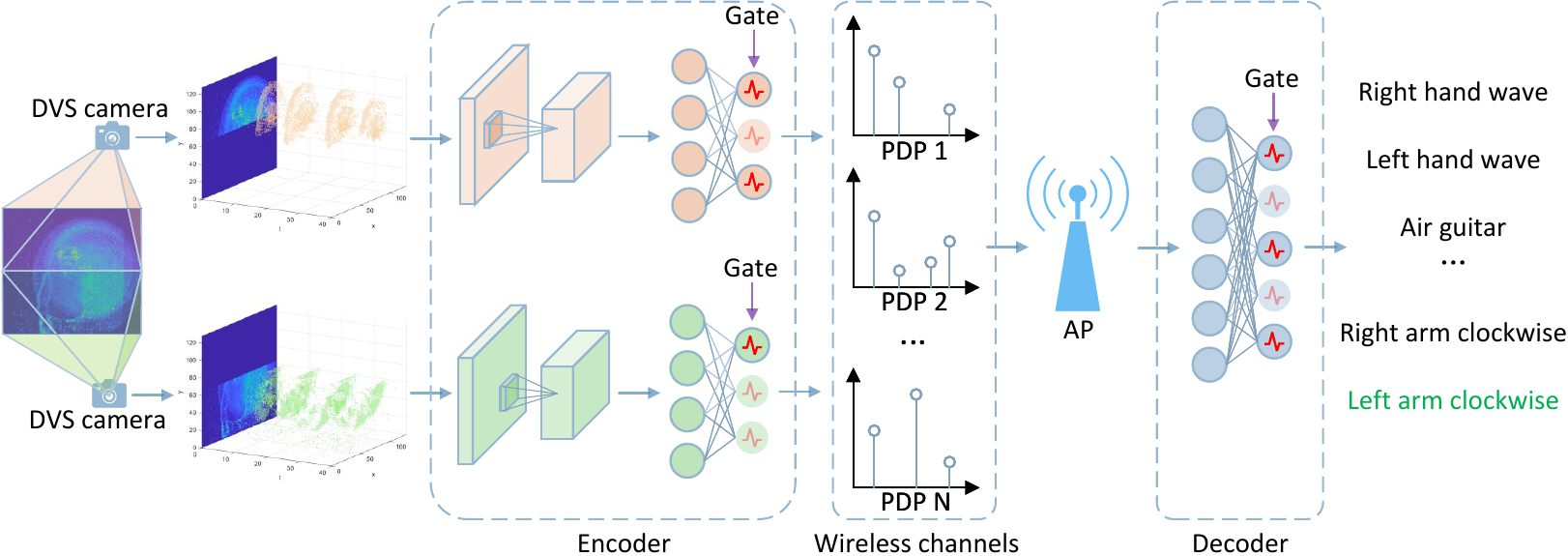}
		\caption{}\label{continous_DVSG}
	\end{subfigure}
	
	\begin{subfigure}[b]{0.32\textwidth}
		\centering
		\includegraphics[width=\textwidth]{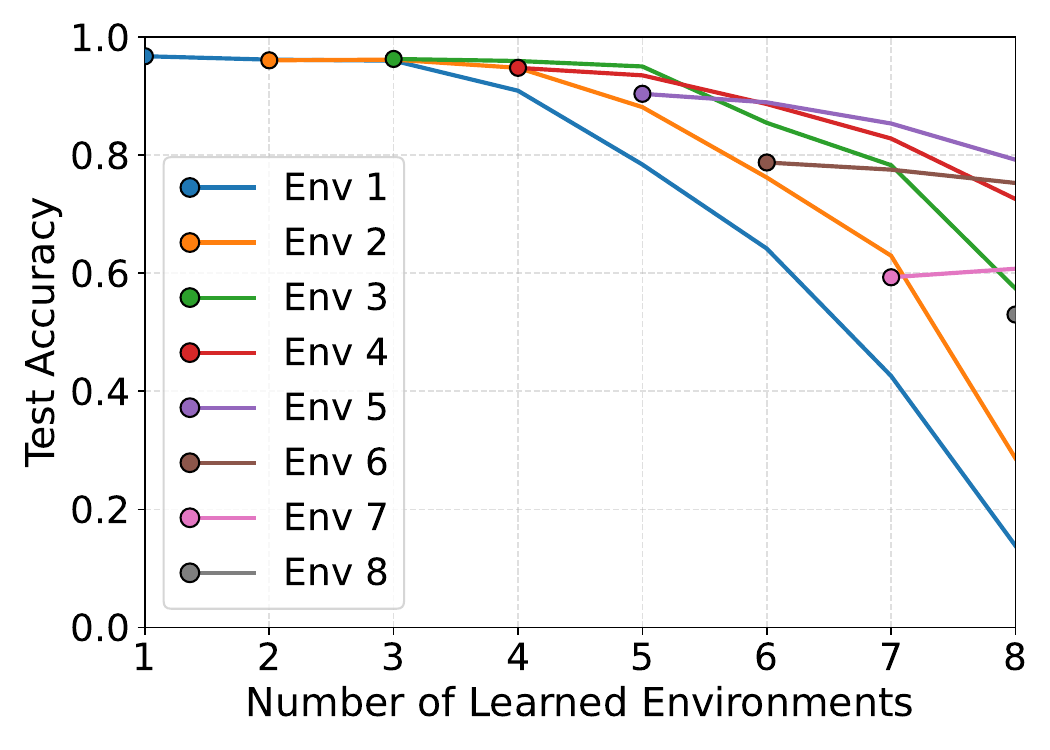}
		\caption{}
		\label{fig:2a}
	\end{subfigure}
	\hfill % Creates horizontal space between the figures
	\begin{subfigure}[b]{0.32\textwidth}
		\centering
		\includegraphics[width=\textwidth]{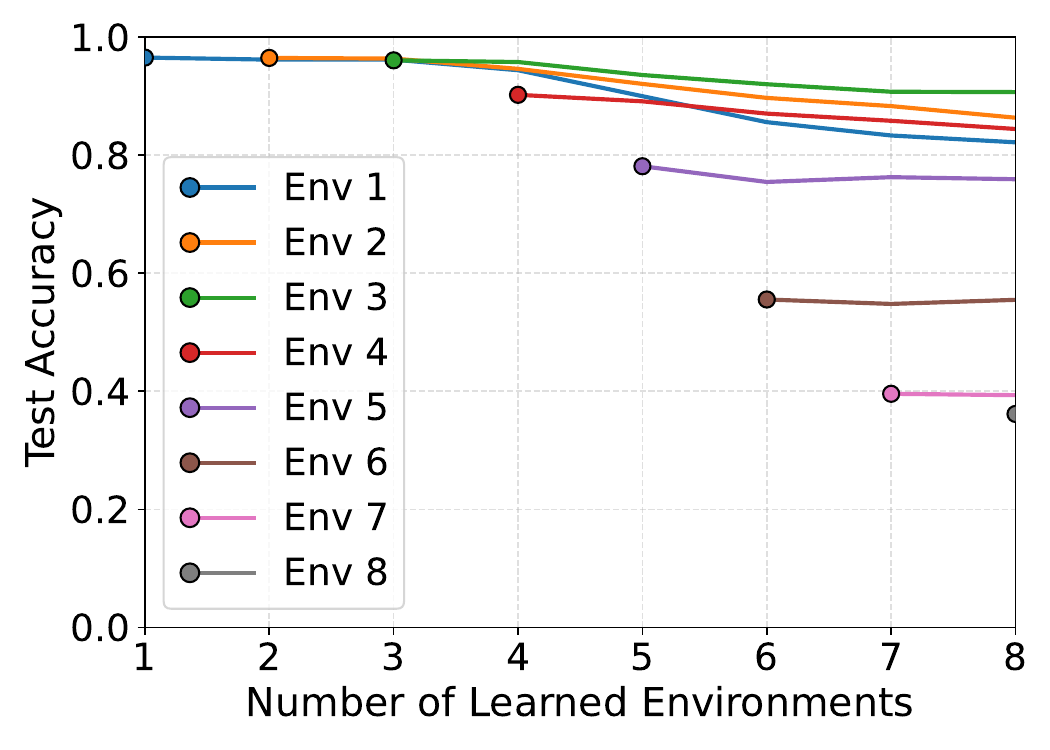}
		\caption{}
		\label{fig:2b}
	\end{subfigure}
	\hfill
	\begin{subfigure}[b]{0.32\textwidth}
		\centering
		\includegraphics[width=\textwidth]{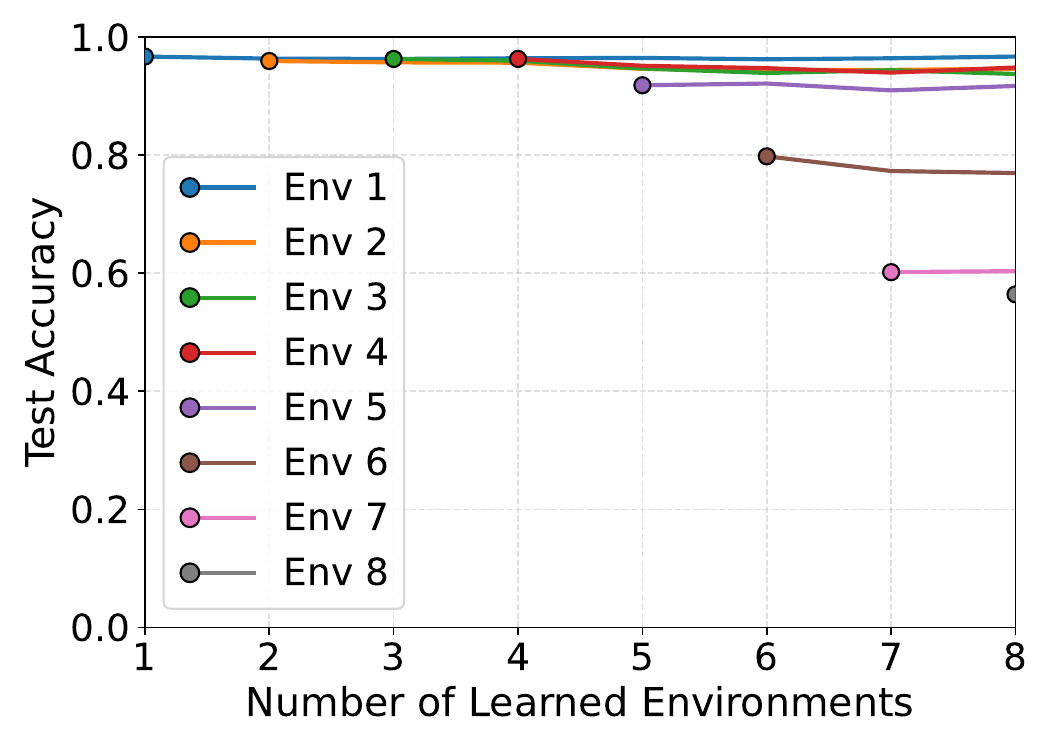}
		\caption{}
		\label{fig:2c}
	\end{subfigure}
	
	\begin{subfigure}[b]{0.32\textwidth}
		\centering
		\includegraphics[width=\textwidth]{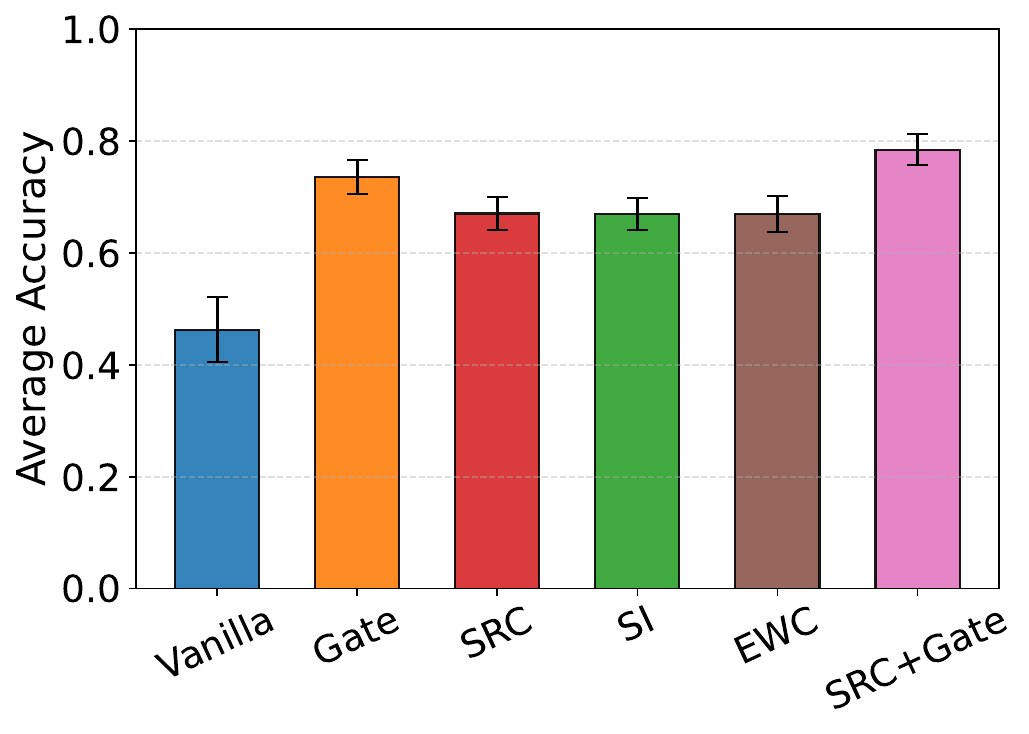}
		\caption{}
		\label{fig:2d}
	\end{subfigure}
	\hfill
		\begin{subfigure}[b]{0.32\textwidth}
		\centering
		\includegraphics[width=\textwidth]{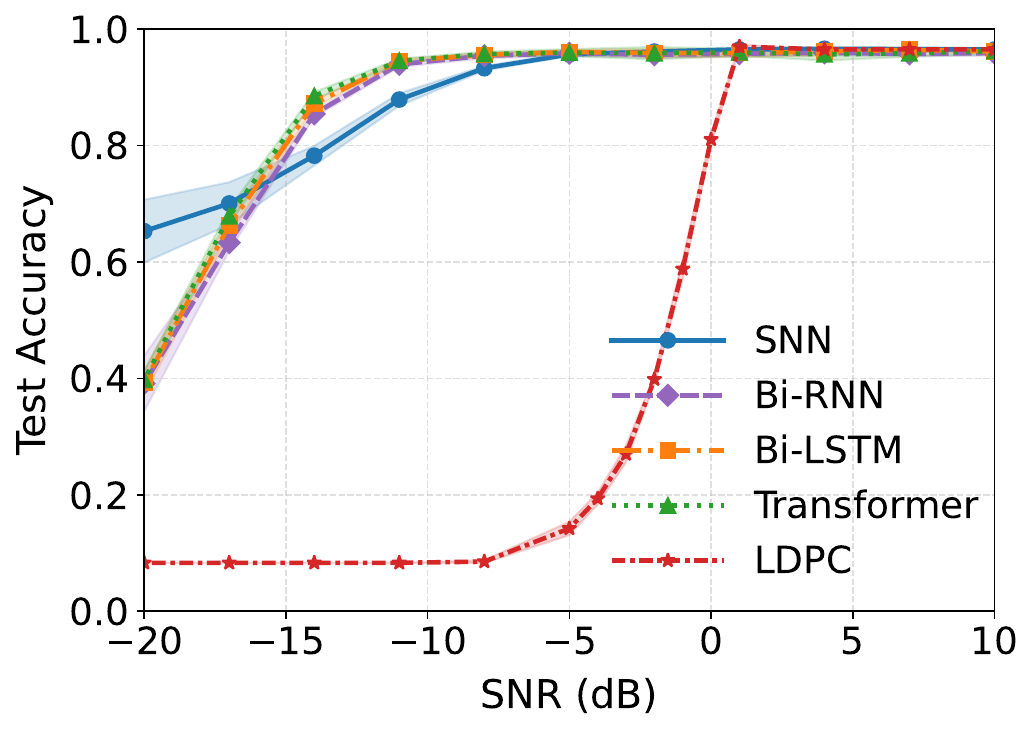}
		\caption{}
		\label{fig:2e}
	\end{subfigure}
	\hfill
	\begin{subfigure}[b]{0.32\textwidth}
		\centering
		\includegraphics[width=\textwidth]{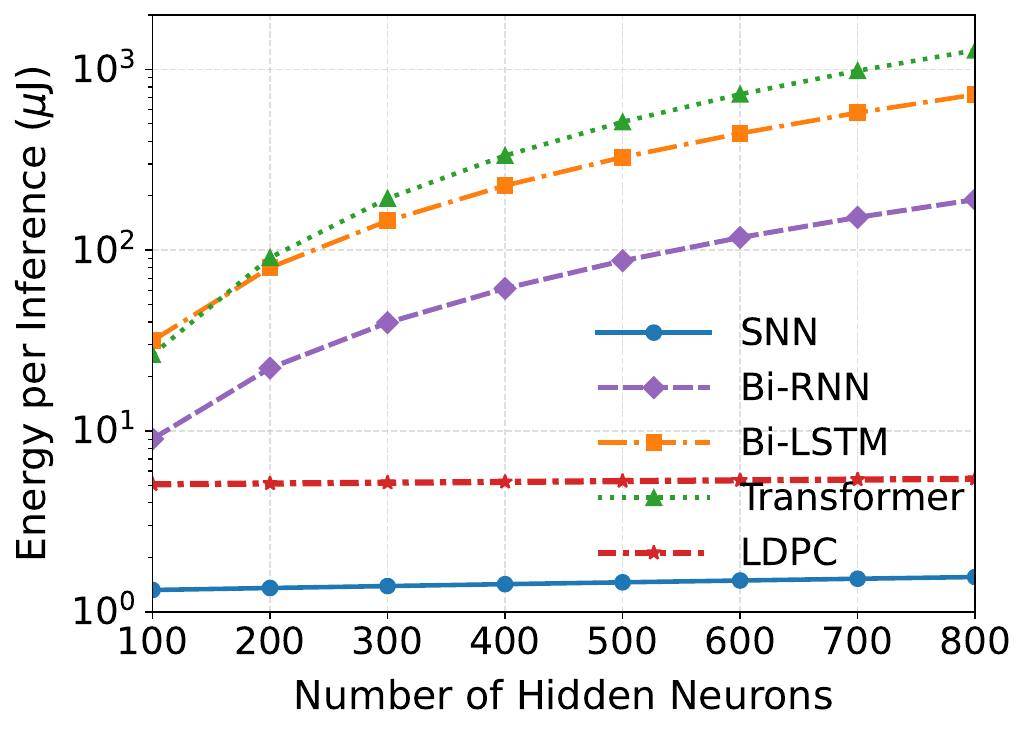}
		\caption{}
		\label{fig:2f}
	\end{subfigure}
	\vspace{-0.4em}

	\caption{\linespread{1.0}\selectfont \textbf{Performance evaluation of SpikACom on neuromorphic semantic communication.} 
		\textbf{(a)} System overview: Event data collected by distributed dynamic vision sensor (DVS) cameras is encoded by SNNs and transmitted over varying wireless channels with different power-delay profiles (PDPs) to the access point (AP). 
		\textbf{(b)}--\textbf{(d)} Evolution of test accuracy versus the number of learned environments. Each curve tracks the test accuracy of a specific environment $i$ evaluated after the model has learned up to environment $k$ (where $k \ge i$).  Comparisons include \textbf{(b)} Vanilla fine-tuning, \textbf{(c)} EWC, \textbf{(d)} SpikACom (ours). The environments are arranged in descending SNR, ranging from $8$~dB (Env 1) to $-20$~dB (Env 8).
		\textbf{(e)} Ablation study of SpikACom. The y-axis represents the average accuracy across all environments after the sequential learning process is completed. 
		\textbf{(f)} Performance comparison of SNN against ANN baselines (Bi-RNN, Bi-LSTM, Transformer) and traditional LDPC coding under different SNRs. 
		\textbf{(g)} Inference energy comparison. The energy consumption of SNN exhibits negligible scaling with the network size compared with ANNs.}
	\label{fig:DVSG}
\end{figure}

Conventional communication systems are designed to reconstruct the transmitted bitstream as accurately as possible, with performance typically measured by bit-error rate. Recent advances in deep learning have made it feasible to extract meaningful semantic information directly from raw data. This shift has given rise to semantic communication \cite{semantic1}, in which the goal is no longer to recover every bit but to convey the task-relevant information. By transmitting only the semantic features that influence the final task, semantic communication can eliminate a large portion of redundant data.
Leveraging deep learning-based joint source-channel coding, this approach significantly reduces transmission overhead and improves robustness to channel impairments \cite{semantic2,semantic3}.

Despite these advantages, practical semantic communication systems remain constrained by the prohibitive computational overhead at the wireless edge. Extracting high-level semantics in ANN-based architectures typically involves dense floating-point matrix operations, making such approaches difficult to deploy on resource-limited devices.
To address this challenge, we introduce an SNN-based framework tailored for neuromorphic semantic encoding and transmission. 
We consider a representative distributed neuromorphic sensing scenario (Fig.~\ref{continous_DVSG}), mirroring future edge applications in which multiple neuromorphic sensors capture and transmit sparse, event-driven data. The DVS128 Gesture dataset~\cite{dvsg_dataset}, a widely used event-based recognition benchmark, is employed to evaluate semantic encoding, transmission, decoding, and continual adaptation under realistic spiking activity.

We first examine the stability of SpikACom against catastrophic forgetting.
Fig.~\ref{fig:2a}–\ref{fig:2c} show the performance trajectories of the evaluated schemes.
The vanilla fine-tuning approach (Fig.~\ref{fig:2a}) lacks mechanisms to prevent catastrophic forgetting, leading to a sharp decline in inference accuracy when exposed to new environments. The classic regularization based method
EWC (Fig.~\ref{fig:2b}), slows down this degradation, but the accuracy still decreases as the channel distribution shifts.
For example, the accuracy in the initial environment falls by about $10\%$ after the model has traversed eight different environments. Moreover, the cumulative constraints imposed by EWC progressively restrict the parameter space, leading to reduced accuracy in later environments.
In contrast, the proposed method (Fig.~\ref{fig:2c}) prevents forgetting while simultaneously maintaining competitive performance in new environments. This improvement arises because SpikACom explicitly exploits the relationships among channel distributions, enabling more efficient allocation of plasticity while preserving stability.

Fig.~\ref{fig:2d} shows the ablation study of the proposed SpikACom. From the figure, both the hypernet-based gating mechanism and the SRC module significantly improve the final accuracy compared with vanilla fine-tuning.
In addition, the SRC method achieves performance very close to gradient-based regularization approaches such as EWC and synaptic intelligence (SI), validating that spike-based statistics can serve as effective proxies for second-order gradient information.
Fig.~S1 in the Supplementary Materials further compares the spiking rate vector with the principal eigenvector of the Fisher information matrix (FIM), revealing a noticeable alignment between the two. However, unlike EWC or SI, SRC incurs much lower computational and memory cost, making it more practical for wireless communication scenarios. Finally,
Combining the context gate with SRC yields the best final average accuracy of $0.7849\pm 0.0279$, indicating that the method effectively mitigates catastrophic forgetting.

Fig.~\ref{fig:2e} compares the SNN with several ANN architectures, including a bidirectional recurrent neural network (Bi-RNN) \cite{bi-RNN}, a bidirectional long short-term memory (Bi-LSTM) \cite{LSTM}, a Transformer \cite{transformer}, and a conventional bit-level low-density parity-check (LDPC) \cite{ldpc} based scheme.
At relatively high signal-to-noise ratio (SNR) values (greater than $-5$ dB), the SNN achieves comparable performance to Bi-RNN, Bi-LSTM, and Transformer.
As the SNR decreases to the range of $-5$ dB to $-16$ dB, the SNN shows a small performance gap relative to ANN-based models, reflecting the increased difficulty of optimization under strong noise.
However, when the SNR drops further, the SNN begins to outperform the ANN architectures. This improvement comes from the sparse and binary nature of spike-based joint source-channel coding, which provides stronger robustness against noise compared with floating-point ANN features.
The LDPC-based scheme, which transmits raw spike data to the access point (AP) and then classifies them using a lightweight SNN, suffers a large performance drop around 0 dB and is significantly inferior to deep learning-based joint source-channel coding.

In terms of energy consumption (Fig.~\ref{fig:2f}), SNNs consume roughly an order of magnitude less energy than Bi-LSTM, Bi-RNN, and Transformer models. Moreover, the inference energy of SNNs scales only weakly with network size, because SNNs are event-driven and depend primarily on the number of emitted spikes, which is governed by the task’s semantic complexity and the channel condition.
In contrast, ANN-based models require fixed floating-point operations. Consequently, their energy consumption increases rapidly as the model size grows.
We also observe that SNN-based joint source-channel coding is more energy-efficient than LDPC-based schemes, providing further evidence for the high energy efficiency of SNNs.

We further investigate the convergence behavior of SpikACom (Fig.~S2) in comparison with an orthogonal gate modulation baseline (Fig.~S3). The results show that although orthogonal gate modulation reduces interference by isolating neuron groups, it restricts knowledge transfer across environments, leading to slow adaptation. In contrast, SpikACom exploits cross-environment correlations to provide better initialization, resulting in faster convergence and more reliable performance under challenging channel conditions.

\subsection*{MIMO Beamforming}

\begin{figure}
	\centering

	\begin{subfigure}[b]{0.96\textwidth}
		\centering
		\includegraphics[width=\textwidth]{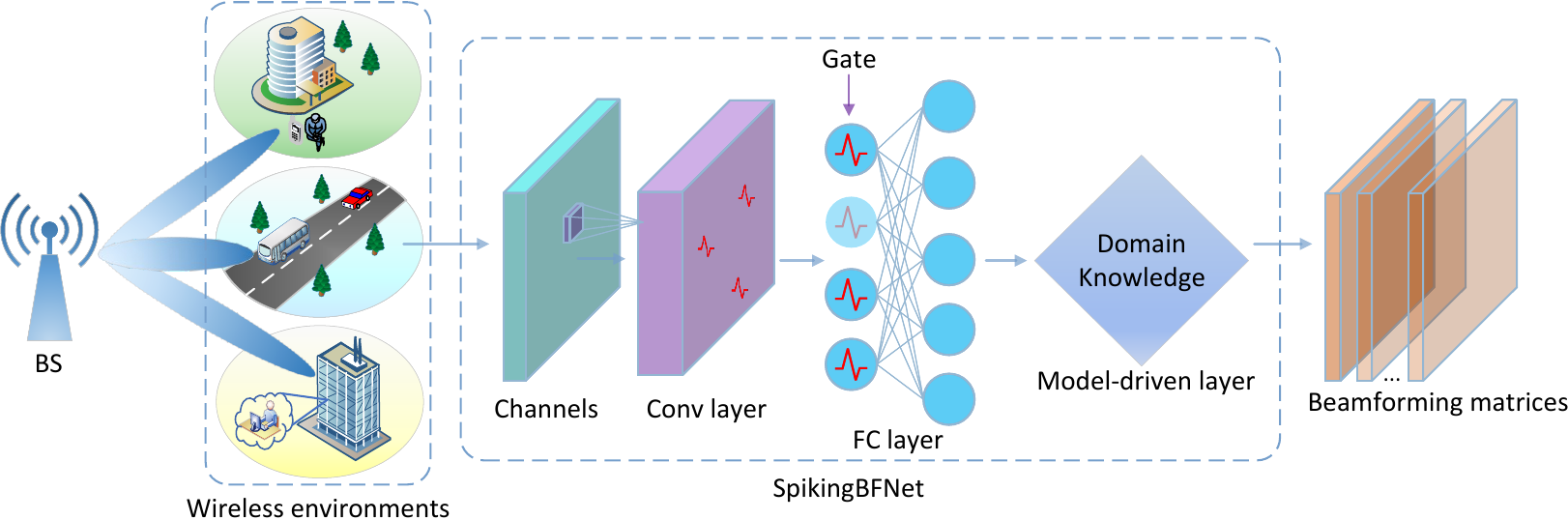}
		\caption{}\label{continuous_beamforming}
	\end{subfigure}
	
	\begin{subfigure}[b]{0.32\textwidth}
		\centering
		\includegraphics[width=\textwidth]{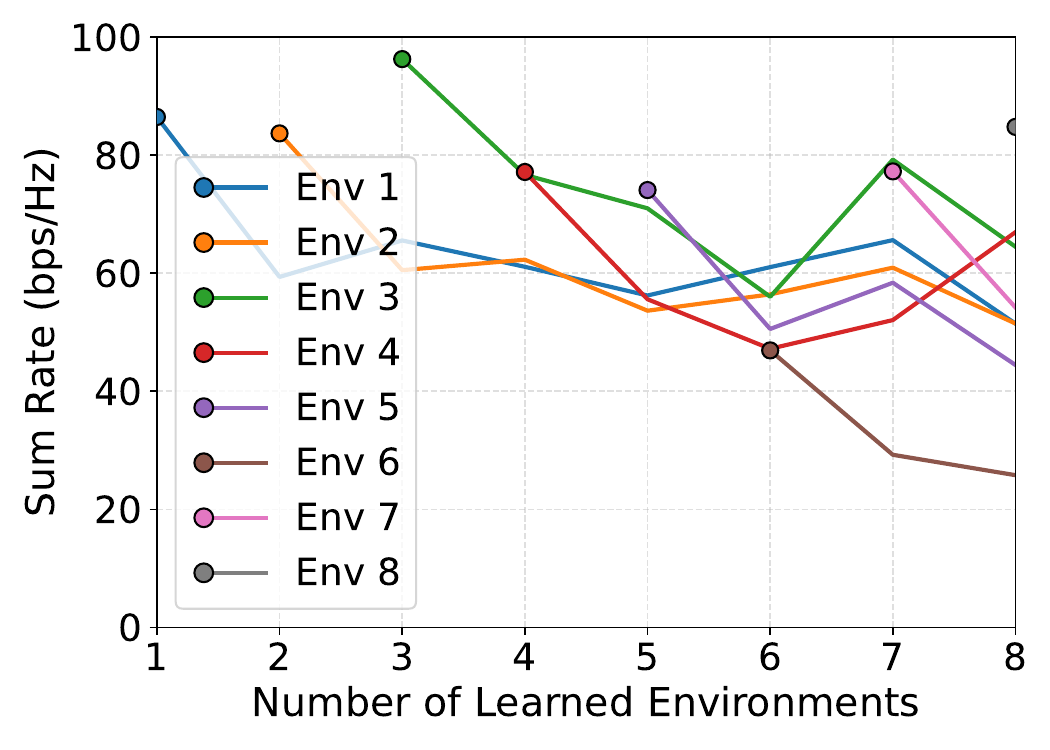}
		\caption{}
		\label{fig:3a}
	\end{subfigure}
	\hfill % Creates horizontal space between the figures
	\begin{subfigure}[b]{0.32\textwidth}
		\centering
		\includegraphics[width=\textwidth]{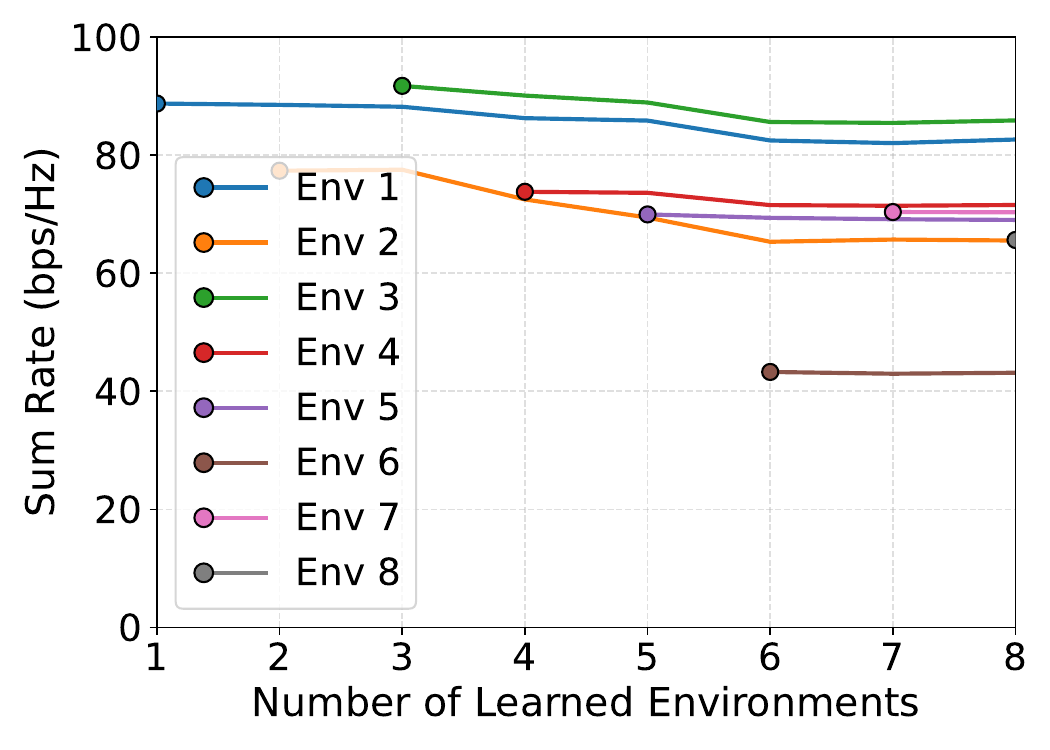}
		\caption{}
		\label{fig:3b}
	\end{subfigure}
	\hfill
	\begin{subfigure}[b]{0.32\textwidth}
		\centering
		\includegraphics[width=\textwidth]{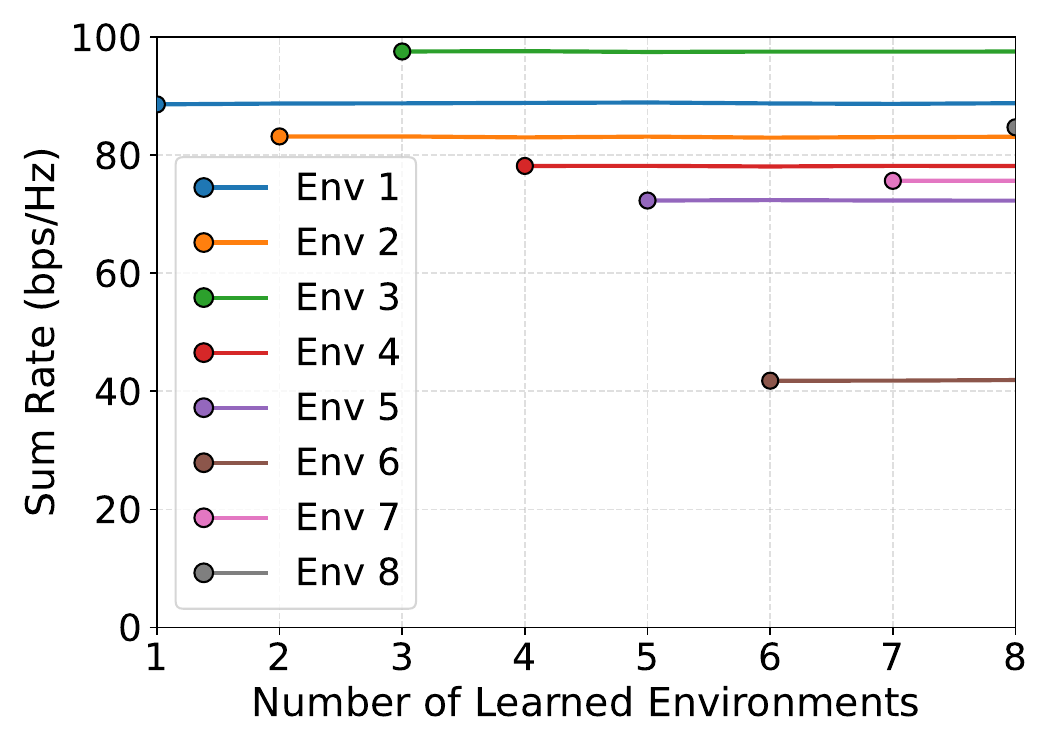}
		\caption{}
		\label{fig:3c}
	\end{subfigure}
	
	\begin{subfigure}[b]{0.32\textwidth}
		\centering
		\includegraphics[width=\textwidth]{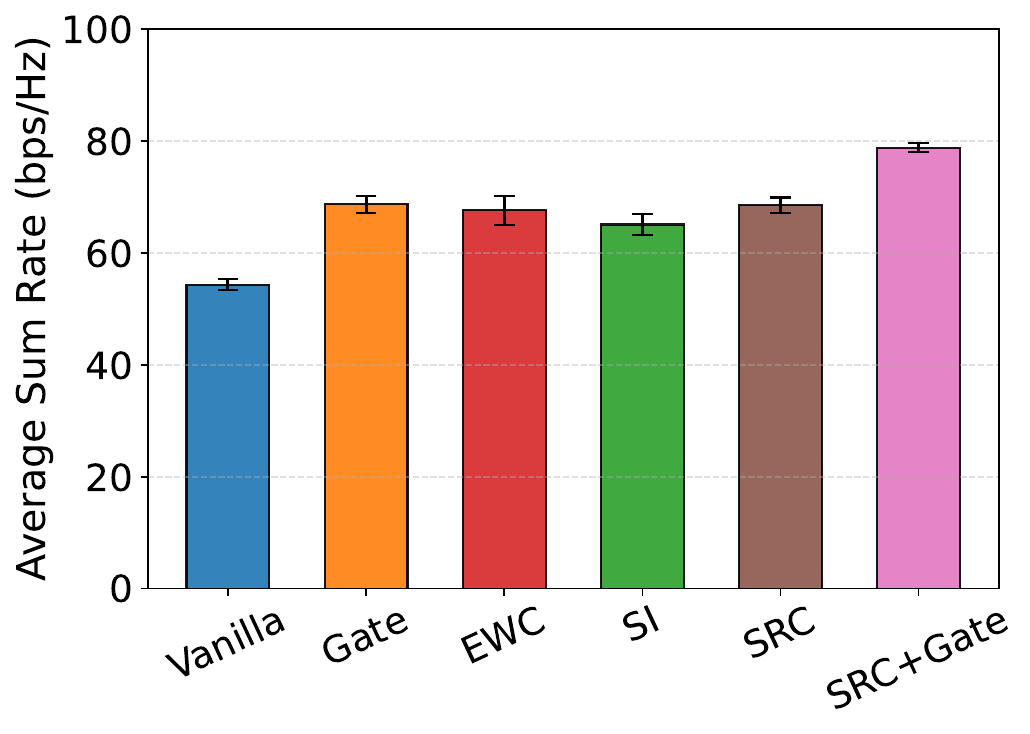}
		\caption{}
		\label{fig:3d}
	\end{subfigure}
	\hfill
	\begin{subfigure}[b]{0.32\textwidth}
		\centering
		\includegraphics[width=\textwidth]{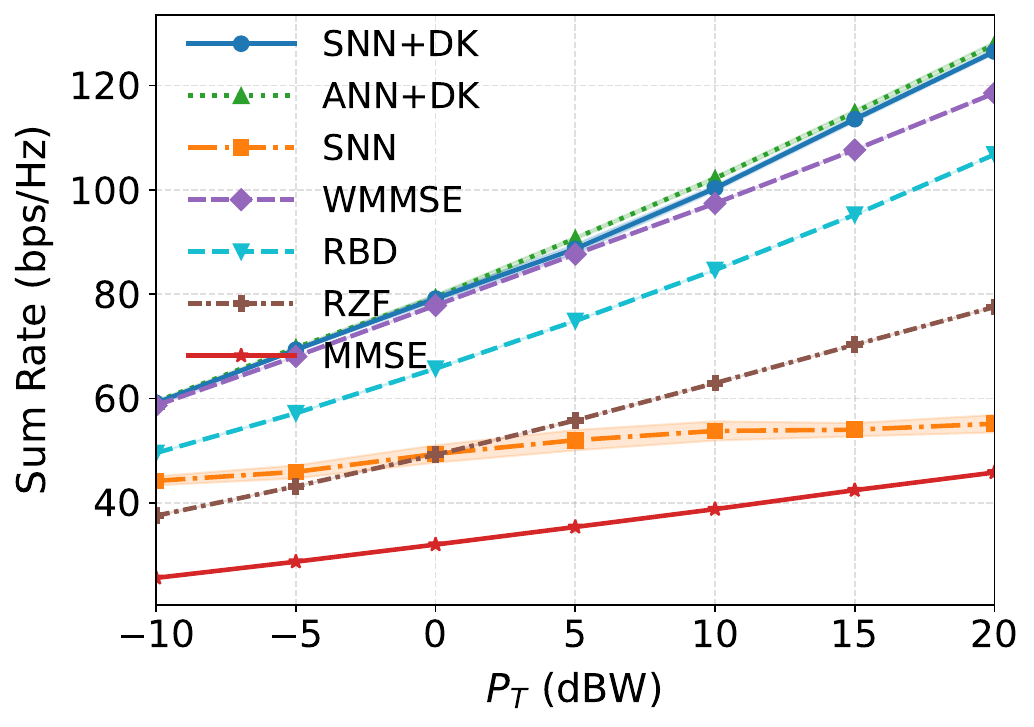}
		\caption{}
		\label{fig:3e}
	\end{subfigure}
	\hfill
	\begin{subfigure}[b]{0.32\textwidth}
		\centering
		\includegraphics[width=\textwidth]{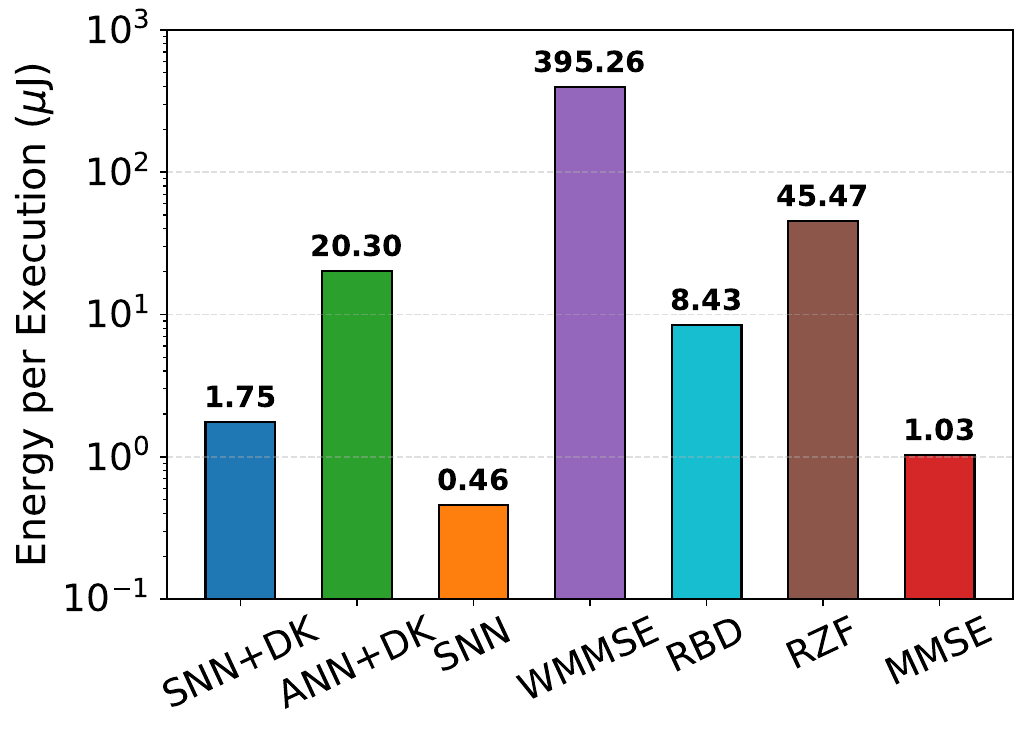}
		\caption{}
		\label{fig:3f}
	\end{subfigure}
	
	\caption{\linespread{1.0}\selectfont \textbf{Performance evaluation of SpikACom on multi-user MIMO beamforming.}
	\textbf{(a)} System overview: A base station (BS) forms beams serving multiple users located in different wireless environments. The complex-valued channel state information is encoded into spikes for energy-efficient processing, and a model-driven module is integrated to improve beamforming performance.
	\textbf{(b)}--\textbf{(d)} Sum-rate evolution over sequential environments. Comparisons include \textbf{(b)} Vanilla fine-tuning, \textbf{(c)} EWC, and \textbf{(d)} SpikACom (ours). The channels are generated based on the 3D ray-tracing DeepMIMO dataset, where each environment corresponds to a distinct user location distribution. 
	\textbf{(e)} Ablation study. The y-axis shows the average sum rate across all environments after the sequential learning process is completed.
	\textbf{(f)} Sum-rate performance versus the maximum transmit power $P_{T}$. We compare the proposed SNN + domain knowledge (DK) against ANN + DK, SNN without DK, and traditional beamforming benchmarks (WMMSE, RBD, RZF, and MMSE).
	\textbf{(g)} Energy consumption comparison of analyzed schemes for beamforming. The energy consumption of SNN+DK is reduced by over $10\times$ compared with the ANN benchmark and is comparable to low-complexity traditional baselines (e.g. MMSE).}
	\label{fig:Beamforming}
\end{figure}

Beamforming is a fundamental signal processing function in multi-antenna wireless systems and plays a central role in 4G, 5G, and the emerging 6G architectures \cite{MIMO_overview}.
Its objective is to steer the transmitted energy toward the intended user by properly weighting the signals across antennas, thereby improving link reliability, extending coverage, and suppressing interference.
Early beamforming methods were mostly model-driven, including matched filtering, zero-forcing (ZF), minimum mean-squared error (MMSE), block diagonalization (BD) \cite{BD}, regularized BD (RBD) \cite{RBD}, and weighted MMSE (WMMSE) \cite{WMMSE}.
These techniques rely on accurate CSI and form the basis of most practical MIMO precoding designs.
More recently, data-driven deep learning approaches have been applied to beamforming optimization \cite{Bf_unfolding1,Bf_unfolding2,Bf_unfolding3,Bf_unfolding4}, leveraging neural networks to improve performance and reduce algorithmic complexity, and have demonstrated empirical effectiveness.
However, these ANN-based algorithms still require massive floating-point matrix operations.
In addition, efficient mechanism to deploy ANN-based beamforming in dynamic wireless environments remain scarce.
Therefore, as illustrated in Fig.~\ref{continuous_beamforming}, we employ the proposed SpikACom framework for beamforming in continually changing environments, with the system setting detailed in the Materials and Methods section.
Moreover, to further enhance the performance and interpretability of SNN-based beamforming, we incorporate domain knowledge into the design of the backbone SNN (see Materials and Methods), combining the high performance of classical model-driven designs with the energy efficiency of SNNs.

We first examine the performance of the proposed framework in mitigating catastrophic forgetting.
From Fig.~\ref{fig:3a} to Fig.~\ref{fig:3c}, the forgetting behavior is progressively reduced when moving from vanilla fine-tuning to EWC, and from EWC to the proposed SpikACom.
The ablation study in Fig.~\ref{fig:3d} shows a trend consistent with the observations in the semantic communication task.
In addition, we note that in this more complex large-scale matrix optimization problem, SRC exhibits more stable performance, achieving an average of $68.52 \pm 1.39$ bps/Hz compared with $67.61 \pm 2.66$ bps/Hz for EWC and $65.09 \pm 1.95$ bps/Hz for SI.
This robustness can be attributed to the numerical stability: matrix operations such as inversion often suffer from poor conditioning and can lead to gradient explosion during backpropagation, which adversely affects gradient-based regularizers like EWC or SI.
In contrast, SRC relies on spiking rate naturally bounded between $0$ and $1$, making it less sensitive to numerical issues.

Fig.~\ref{fig:3e} compares SNN+domain knowledge (DK) beamforming with ANN+DK, vanilla SNN without the model-driven layer, and several classical model-based methods.
Across the entire transmit power range from $-10$ dBW to $20$ dBW, SNN+DK achieves a sum-rate performance that closely matches ANN+DK.
WMMSE performs comparably at low transmit power, but the gap widens as the transmit power increases.
This occurs because WMMSE relies on block coordinate descent, which may become trapped in local optima under ill-conditioned or unbalanced channel realizations.
In contrast, data-driven methods handle the non-convex optimization landscape more effectively, thus achieving higher performance.
Other model-based algorithms, including RBD, RZF, and MMSE, show a noticeable performance gap compared to data-driven approaches.
The vanilla SNN, which lacks guidance from domain knowledge, suffers from significant performance degradation and exhibits unstable convergence. This highlights the importance of integrating model information into SNN architectures in complex signal processing tasks.

Fig.~\ref{fig:3f} reports the energy consumption of different beamforming algorithms. SNN+DK consumes only $1.75\, \mu$J per execution, which is an order of magnitude lower than the $20.30\, \mu$J required by ANN+DK and is comparable to the lightweight MMSE baseline ($1.03\, \mu$J).
Pure SNN achieves the lowest energy consumption at $0.46\,\mu$J. However, its insufficient performance limits its practical utility. Therefore,
SNN+DK offers a more favorable performance-energy trade-off.
RBD and RZF consume energy levels similar to ANN+DK. The optimization-based WMMSE method is the most expensive algorithm, consuming about $400\,\mu$J per execution due to the large number of iterations required.

\subsection*{Channel Estimation}\label{sec7}
\begin{figure}
	\centering
	\begin{subfigure}[b]{0.96\textwidth}
		\centering
		\includegraphics[width=\textwidth]{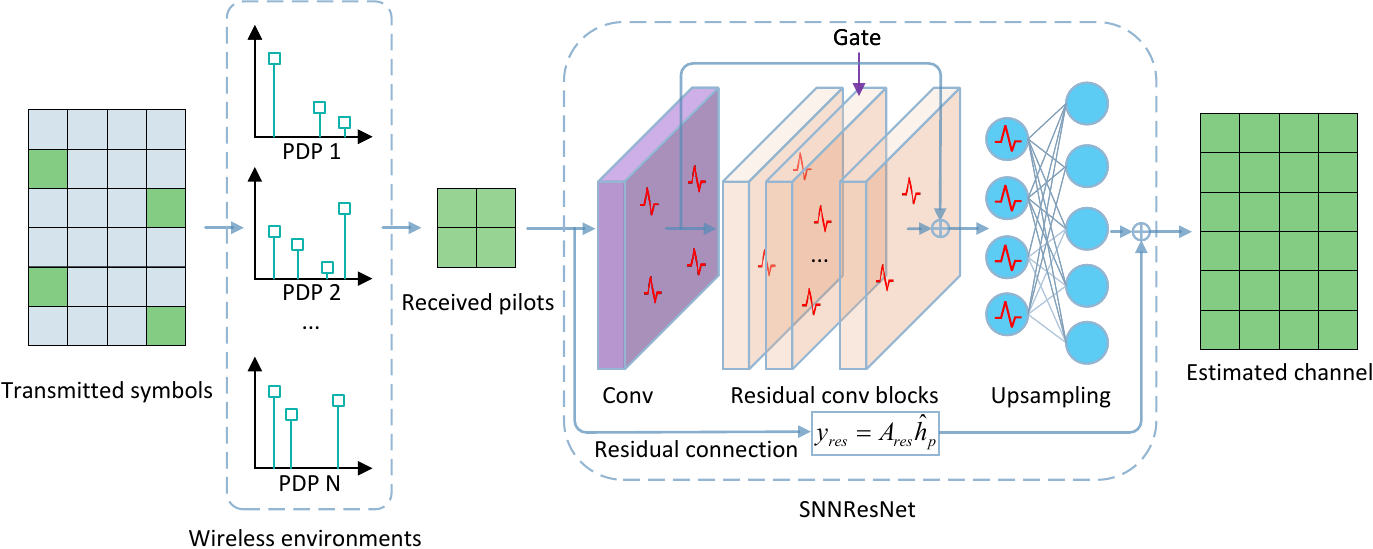}
		\caption{}\label{channel_estimation}
	\end{subfigure}
	
	\begin{subfigure}[b]{0.32\textwidth}
		\centering
		\includegraphics[width=\textwidth]{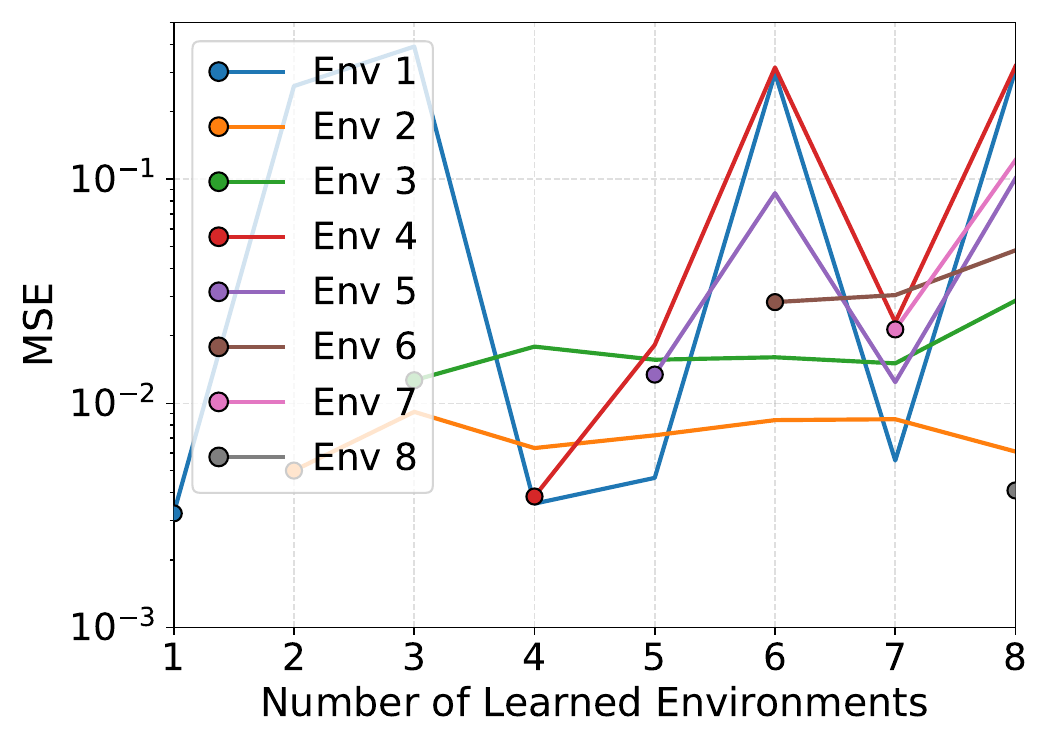}
		\caption{}
		\label{fig:4a}
	\end{subfigure}
	\hfill % Creates horizontal space between the figures
	\begin{subfigure}[b]{0.32\textwidth}
		\centering
		\includegraphics[width=\textwidth]{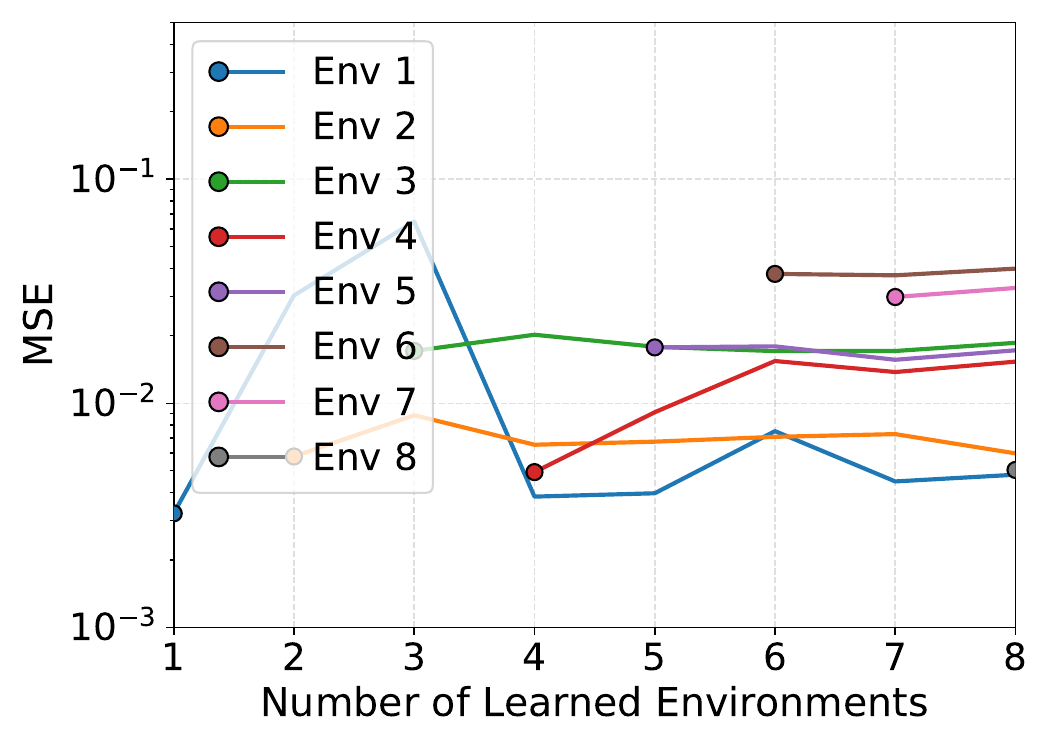}
		\caption{}
		\label{fig:4b}
	\end{subfigure}
	\hfill
	\begin{subfigure}[b]{0.32\textwidth}
		\centering
		\includegraphics[width=\textwidth]{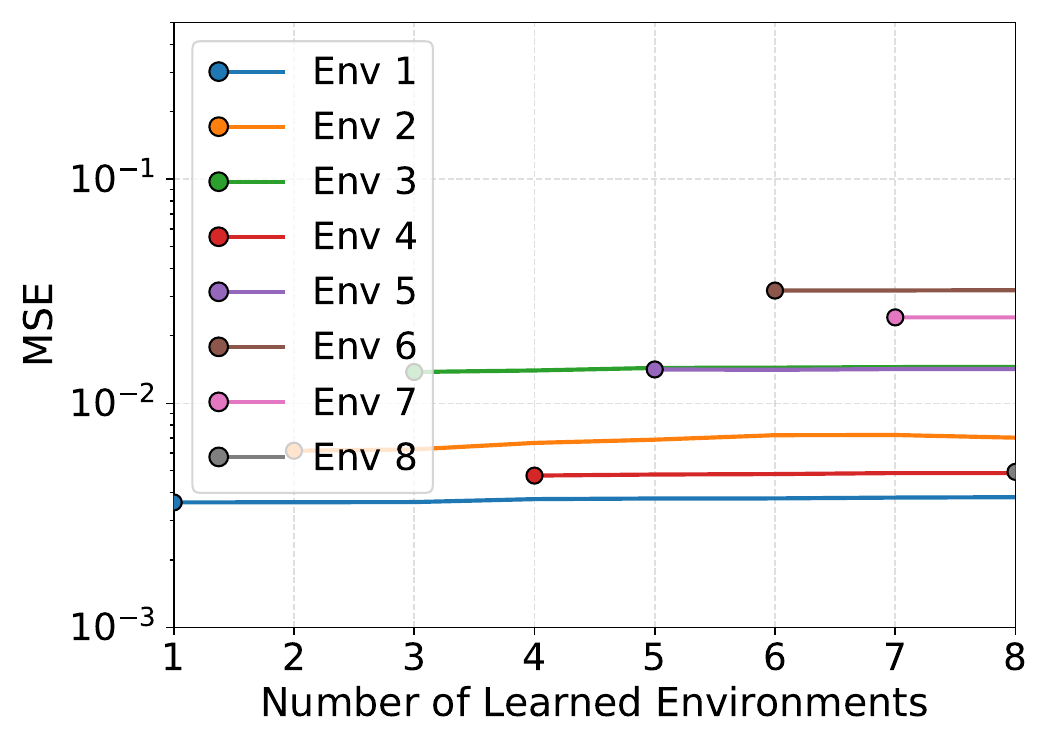}
		\caption{}
		\label{fig:4c}
	\end{subfigure}
	
	\begin{subfigure}[b]{0.32\textwidth}
		\centering
		\includegraphics[width=\textwidth]{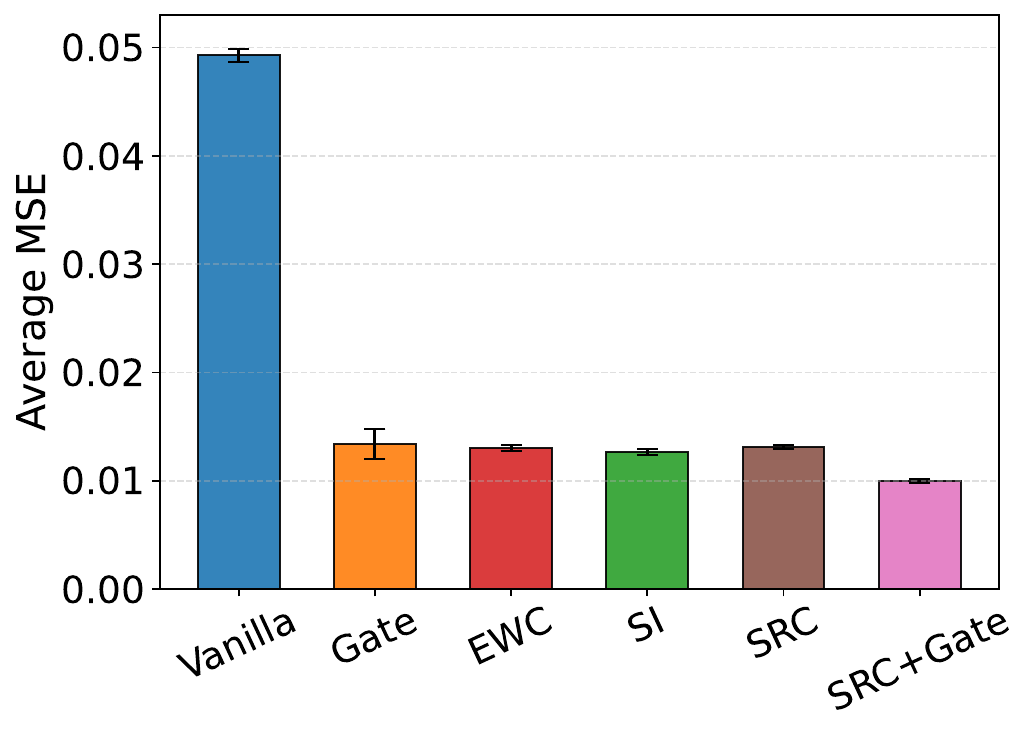}
		\caption{}
		\label{fig:4d}
	\end{subfigure}
	\hfill
	\begin{subfigure}[b]{0.32\textwidth}
		\centering
		\includegraphics[width=\textwidth]{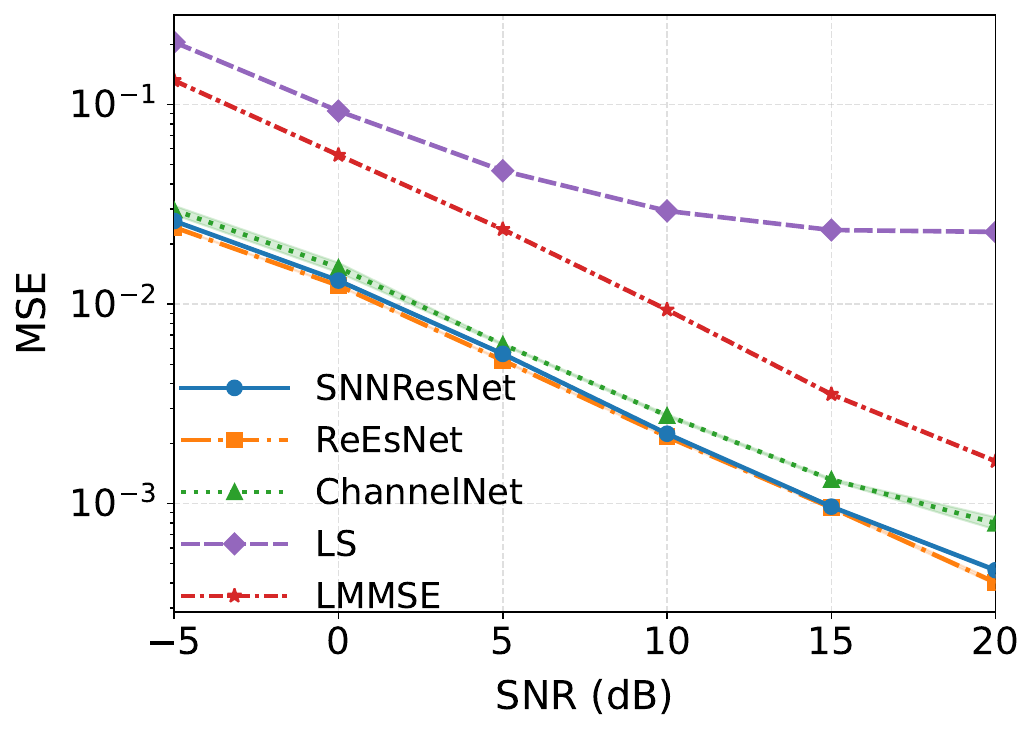}
		\caption{}
		\label{fig:4e}
	\end{subfigure}
	\hfill
	\begin{subfigure}[b]{0.32\textwidth}
		\centering
		\includegraphics[width=\textwidth]{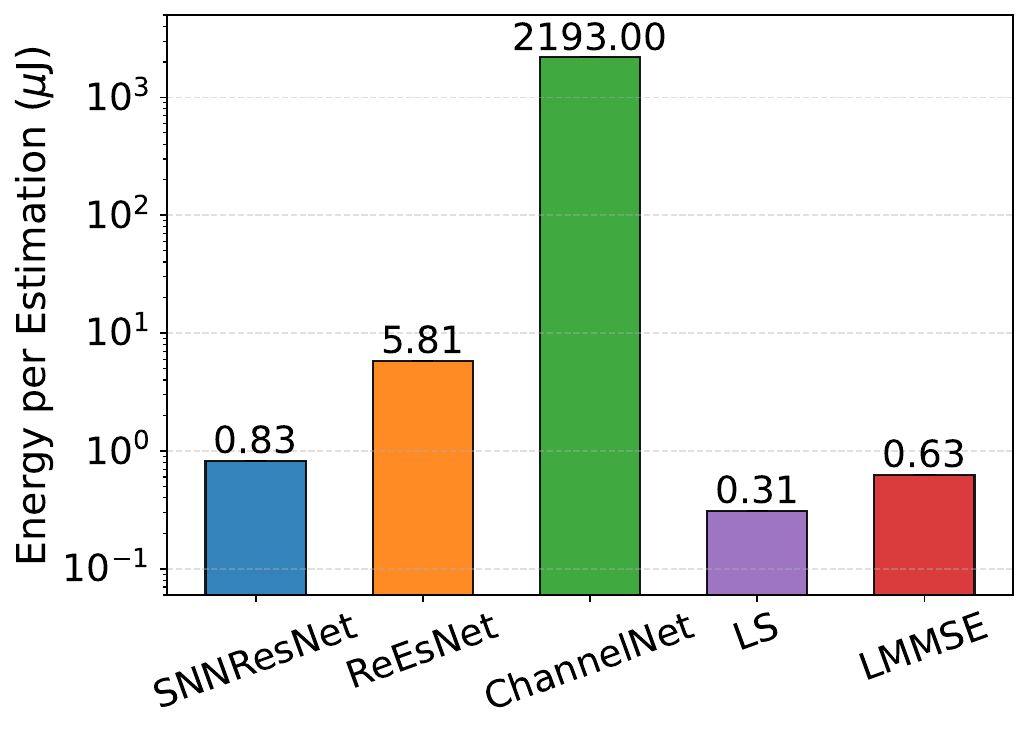}
		\caption{}
		\label{fig:4f}
	\end{subfigure}	

	\vspace{-1em}
	\caption{\linespread{1.0}\selectfont \textbf{Performance evaluation of SpikACom on channel estimation in OFDM systems.} 
		\textbf{(a)} System overview: Sparse pilots are inserted into the OFDM frequency-time grid and transmitted through dynamic wireless channels. The received pilots are processed by the proposed SNNResNet, which uses spike-based residual convolutional blocks for feature extraction and incorporates an end-to-end LMMSE-inspired residual connection to facilitate channel reconstruction.
		\textbf{(b)}--\textbf{(d)} MSE evolution over sequential environments. Comparisons include \textbf{(b)} Vanilla fine-tuning, \textbf{(c)} EWC, and \textbf{(d)} SpikACom (ours). The WINNER II channel model is adopted, where each environment corresponds to a distinct propagation scenario. Lower MSE indicates better estimation performance.
		\textbf{(e)} Ablation study. The y-axis represents the average MSE, computed using the arithmetic mean, across all environments after the sequential learning process is completed. 
		\textbf{(f)} Performance comparison among SNN-based channel estimation (SNNResNet), ANN-based channel estimation (ChannelNet and ReEsNet), and traditional estimators (LS and LMMSE). \textbf{(g)} Energy consumption comparison. SNN-based scheme is  $7\times$ more energy-efficient than lightweight ANN baseline (ReEsNet).}		
	\label{fig:OFDMCE}
\end{figure}
CSI is indispensable for reliable wireless communications. It serves as the prerequisite for many key communication functions, including signal detection, beamforming, and resource allocation. Inaccurate CSI leads to mismatched processing and noise enhancement, causing substantial performance loss in terms of reliability and spectral efficiency. In this work, we focus on channel estimation in OFDM systems. OFDM is a foundational physical-layer modulation technology and underpins many of the most influential communication standards, including 4G long-term evolution (LTE), 5G new radio (NR), and the IEEE 802.11 Wi-Fi family. However, obtaining accurate CSI in OFDM systems is a non-trivial problem due to complex environmental scattering and limited pilot observations.  
Traditional channel estimators such as least squares (LS) and MMSE perform well under simplified channel models but degrade in the presence of hardware impairments, non-ideal propagation, or rich scattering.
Recent deep learning-based estimators mitigate some of these issues by capturing complex channel characteristics directly from data.
However, their computational cost incurs substantial energy overhead, challenging the sustainability of next-generation wireless communication systems.
To overcome these challenges, we develop an SNN-based OFDM channel estimation method within the SpikACom framework, leveraging the energy-efficient computational properties of SNNs while maintaining low estimation error.

We consider an OFDM channel estimation task in a sequentially varying environment (Fig. \ref{channel_estimation}). We adopt the WINNER II channel model \cite{winner2} as the testbed. This channel model encompasses a wide range of realistic propagation scenarios, including indoor/outdoor settings, urban/rural areas, and line-of-sight (LOS)/non-LOS (NLOS) conditions.
Fig. \ref{fig:4a}–Fig. \ref{fig:4c} illustrate the performance trajectories of the compared schemes under the changing environments. In Fig. \ref{fig:4a}, the vanilla fine-tuning approach exhibits significant performance fluctuations due to substantial shifts in the underlying channel distribution.
The EWC method (Fig. \ref{fig:4b}) alleviates this issue to some extent by retaining partial knowledge from past environments. However, it does not explicitly exploit the relationships between previous channel distributions and the current one. Consequently, large distributional gaps still trigger noticeable degradation.
The proposed method (Fig. \ref{fig:4c}) more effectively leverages the relationships between historical and incoming environments. This guidance enables the backbone network to update its parameters in a more informed manner, yielding stable performance across environments.

Fig.~\ref{fig:4d} presents the ablation results of the proposed algorithm.
The fine-tuning baseline performs poorly, yielding a high final mean-squared error (MSE) of $0.049$.
Introducing the hypernet-based gating mechanism substantially improves performance, reducing the MSE to $0.0133$.
Regularization-based approaches achieve similar  results, with final MSEs of $0.0130$ (EWC), $0.0126$ (SI), and $0.0131$ (SRC).
The best performance is obtained by combining the context gate with regularization, which further reduces the MSE to $0.0099$.

Fig.~\ref{fig:4e} compares the SNN-based estimator with the traditional LS and linear MMSE (LMMSE) estimators, as well as two representative ANN-based OFDM channel estimation networks, ChannelNet~\cite{OFDM_CE1} and ReEsNet~\cite{OFDM_CE2}.
ChannelNet decomposes channel estimation into a super-resolution module followed by a denoising module, whereas ReEsNet adopts a lightweight end-to-end residual architecture that serves as a strong ANN baseline.
As shown in Fig.~\ref{fig:4e}, LS performs poorly and its MSE curve saturates in the high-SNR region.
This occurs because LS does not exploit the structural correlation between pilot positions and the whole channel grid, leading to severe degradation, especially when the pilot density is low.
LMMSE achieves better performance than LS but still exhibits a large gap relative to learning-based approaches.
ChannelNet lags behind SNNResNet and ReEsNet, mainly because its two-stage design cannot match the optimization benefits of end-to-end architectures.
ReEsNet and the proposed SNNResNet achieve nearly identical MSE performance and clearly outperform other baselines.

Fig.~\ref{fig:4f} summarizes the energy consumption of different estimators.
SNNResNet requires only $0.827\,\mu\text{J}$ per estimation, which is approximately $7\times$ more energy-efficient than ReEsNet and is comparable to the LMMSE estimator.
ChannelNet is notably energy-intensive, consuming $2.19\,\text{mJ}$ due to its large network size.
LS is the most energy-efficient method, but its estimation accuracy is poor.
Jointly considering both performance and energy efficiency, the SNN-based estimator provides a highly competitive solution for OFDM channel estimation.

\section*{Discussion}\label{sec13}
Modern communication systems face intensifying pressure to meet carbon-neutrality goals as energy consumption in wireless networks continues to escalate. This challenge is exacerbated by the growing reliance on deep learning-based wireless signal processing, which incurs substantial computational overhead. To address this issue, we have explored SNNs as a promising low-power paradigm for signal processing and proposed SpikACom, a framework designed to enable green, adaptive communications. SpikACom tackles the highly dynamic nature of wireless environments via a dual-scale mechanism that integrates (i) context-aware modulation, which exploits similarities in the channel distributions to improve knowledge reuse and mitigate catastrophic forgetting, and (ii) spiking rate consolidation (SRC), a novel mechanism that regularizes the synaptic weights of SNNs based on the firing intensity of spiking neurons, achieving effects comparable to those of traditional gradient-based regularization methods with significantly lower demands on storage and computation.

We evaluated SpikACom on several key tasks critical to current and emerging wireless communications, including neuromorphic semantic communication, multi-user MIMO beamforming, and OFDM channel estimation. Our results show that, with appropriate network architecture design, SNNs can match the performance of conventional ANNs while reducing computational energy by approximately an order of magnitude. Compared with classical signal processing methods, SNNs exhibit superior performance in all considered settings while maintaining energy consumption comparable to low-complexity baselines. Notably, all models were trained directly using backpropagation through time (BPTT) with surrogate gradients, without converting from ANN counterparts, indicating that SNNs possess the intrinsic capability to learn complex wireless signal processing tasks. Furthermore, our findings underscore that integrating domain knowledge into the backbone SNN is essential to fully unlocking this potential. This integration allows the network to capture task-specific physical structures, enhancing both performance and interpretability while introducing only modest increases in complexity.

Overall, our work presents an SNN-based framework for continual learning in dynamic wireless environments and provides evidence that spike-based signal processing can be effective on sophisticated physical-layer tasks. By linking adaptability under distribution shift with neuromorphic efficiency, these results advance the case for sustainable communication architectures toward 6G and beyond. At the same time, translating this promise into robust, real-world deployment demands broad community collaboration. We therefore close by outlining the open challenges and identifying key future directions.
\begin{enumerate}
	\item \textbf{Theoretical guarantees}. Traditional communication methods are model-driven and supported by rigorous theoretical analysis. Although our design incorporates domain knowledge, many SNN-based components are still difficult to analyze and certify compared with classical signal processing blocks. This gap calls for stronger theoretical foundations for SNNs or tighter integration between SNNs and model-based signal processing algorithms. 
	\item \textbf{Hardware implementations}. Current deep learning hardware (e.g., CPUs, GPUs) is mainly based on the von Neumann architecture and is not well matched to event-driven computation. Existing neuromorphic chips (e.g., Loihi \cite{loihi2}, TrueNorth \cite{truenorth}) demonstrate the benefits of SNNs but may not satisfy the strict latency, reliability, and throughput requirements of communication systems. Fully exploiting the energy advantages of SNNs may require communication-aware neuromorphic accelerators and system-level co-design tailored for wireless systems.
	\item \textbf{Training efficiency}. Although SNNs exhibit high energy efficiency during inference, standard SNN training methods that rely on backpropagation through time still incur high memory and computational costs. Thus, more efficient training pipelines that reduce the frequency of weight updates, enable sparse parameter learning, or leverage low-cost biologically inspired local learning rules \cite{STDP} are needed to reduce training overhead. 
	\item \textbf{Standardization efforts}. The reliable use of SNNs in intelligent communication systems depends on unified and well-defined standards. At present, SNN research lacks consistency across neuron models, network architectures, spike-based modulation methods, training protocols, hardware interfaces, and software stacks, leading to fragmented research efforts and duplication of work. A systematic standardization framework is necessary to ensure consistent behavior across platforms and vendors, enable large-scale interoperability, and support the integration of SNN-based modules into future communication infrastructures.
\end{enumerate}

\section*{Materials and Methods}
\subsection*{Spiking Neuron Model}
In this work, we adopt the widely used leaky integrate-and-fire (LIF) neuron model. The dynamics of a LIF neuron can be formulated as
\begin{align}
	\bar{V}[t] &= V[t-1] - \frac{1}{\tau}(V[t-1] - V_{\text{reset}}) + X[t], \\
	S[t] &= \Theta(\bar{V}[t] - V_{\text{th}}), \\
	V[t] &= (1 - S[t])\bar{V}[t] + S[t] V_{\text{reset}},  \label{V_reset}
\end{align}
where $\bar{V}[t]$ represents the pre-reset membrane potential of the spiking neuron at time step $t$, and $\tau$ is the membrane time constant. $V_{\text{reset}}$ denotes the reset membrane potential, and $X[t]$ is the input current. When the membrane potential reaches threshold $V_{\text{th}}$, a spike is generated, denoted by $S[t] \in \{0,1\}$. The function $\Theta(\cdot)$ is the Heaviside step function. Equation \eqref{V_reset} indicates that membrane potential $V[t]$ is reset to $V_{\text{reset}}$ when a spike is triggered.

\subsection*{General Training Setup}
We train the SNNs using the SpikingJelly framework \cite{spikingjelly}. All SNNs are trained using BPTT with the arctangent surrogate gradient \cite{snnsurrogate}. Unless otherwise noted, we use a batch size of 64 and a learning rate of $1 \times 10^{-3}$. We employ the Adam optimizer with a cosine annealing learning rate schedule. By default, $15\%$ of the training samples are reserved for validation. All experiments, including the benchmark evaluations, are repeated 10 times with different random seeds, and the error bars denote the standard deviation across runs.

\subsection*{Hypernet-based Context Modulator}
We employ a hypernet to exploit the relationships between different wireless environments and extract distribution-related information to guide the backbone network. This enables the reuse of previously learned knowledge to accelerate adaptation in similar environments. Meanwhile, for dissimilar channel distributions, the synaptic weights of the backbone network are separated to avoid catastrophic forgetting. This method is primarily inspired by the hybrid modulation network (HMN) \cite{snn_hybrid}, but we introduce several key improvements. First, whereas HMN modulates neuronal firing thresholds via a soft-gating mechanism, we empirically observe that using a binary gate to control the on/off states of spiking neurons yields superior performance \cite{piggyback}. Second, HMN is developed exclusively for the permuted MNIST dataset, which restricts its applicability to other continual learning scenarios. In this work, we adopt a general distance measure to quantify the discrepancy between varying distributions, making our approach not only more effective but also applicable to a broad range of tasks. Moreover, we provide a theoretical analysis of the domain generalization bound under the considered distance metric, offering insights into why context information is necessary to alleviate catastrophic forgetting (see Supplementary Text).

\subsubsection*{Metric for Channel Distribution Distance}
To enable the hypernet to distinguish channel conditions and generate appropriate context-aware gating signals, quantifying the distance between different environments is crucial. There are several widely used metrics for measuring the distribution distance, including f-divergence, total variation distance, and earth mover's distance (EMD). In this paper, we adopt EMD for its favorable continuity property, which helps stabilize neural network training \cite{WGAN}. Specifically, EMD is defined as the minimal cost for transforming one distribution into another:
\begin{equation}
	\text{EMD}(Q, P) = \inf_{\gamma \in \Gamma(Q, P)}\int_{\mathcal{X} \times \mathcal{X}}d(\boldsymbol{x},\boldsymbol{y})\,d\gamma(\boldsymbol{x},\boldsymbol{y}),
\end{equation}
where $\mathcal{X}$ is the sample space, $\Gamma(Q, P)$ denotes the set of all joint distributions whose marginals are $Q$ and $P$, and $d(\boldsymbol{x},\boldsymbol{y})$ represents a distance measure such as the Euclidean norm. 

However, directly computing the EMD involves solving a computationally demanding optimization problem. To address this issue, we propose the Fréchet channel distance (FCD) inspired by the Fréchet Inception distance (FID) \cite{FID}. Specifically, for two Gaussian distributions $\mathcal{N}(\boldsymbol{\mu}_1, \mathbf{\Sigma}_1)$ and $\mathcal{N}(\boldsymbol{\mu}_2, \mathbf{\Sigma}_2)$, their EMD admits a closed-form expression:
\begin{equation}
	\text{EMD}\bigl(\mathcal{N}(\boldsymbol{\mu}_1, \mathbf{\Sigma}_1), \mathcal{N}(\boldsymbol{\mu}_2, \mathbf{\Sigma}_2)\bigr)^2 = \|\boldsymbol{\mu}_1 - \boldsymbol{\mu}_2\|_2^2 + \mathrm{Tr} \bigl(\mathbf{\Sigma}_1 + \mathbf{\Sigma}_2 - 2(\mathbf{\Sigma}_1^{\frac{1}{2}}\mathbf{\Sigma}_2\mathbf{\Sigma}_1^{\frac{1}{2}})^{\frac{1}{2}}\bigr). \label{EMD_gaussian}
\end{equation}
Given two channel distributions $\mathcal{H}_1$ and $\mathcal{H}_2$, we map the samples to a new space through a function $f$ and fit two Gaussian distributions for $f(\mathcal{H}_1)$ and $f(\mathcal{H}_2)$, respectively. Then, the FCD is obtained by substituting the fitted parameters into Eq.~\eqref{EMD_gaussian}. In short, the FCD is formally defined as
\begin{equation}
	\text{FCD}(\mathcal{H}_1, \mathcal{H}_2)^2 \triangleq \|\boldsymbol{\mu}_1 - \boldsymbol{\mu}_2\|_2^2 + \mathrm{Tr}\bigl(\mathbf{\Sigma}_1 + \mathbf{\Sigma}_2 - 2(\mathbf{\Sigma}_1^{\frac{1}{2}}\mathbf{\Sigma}_2\mathbf{\Sigma}_1^{\frac{1}{2}})^{\frac{1}{2}}\bigr), \label{FCD}
\end{equation}
where $\boldsymbol{\mu}_1 = \mathbb{E}_{\boldsymbol{h} \sim \mathcal{H}_1}[f(\boldsymbol{h})]$, $\boldsymbol{\mu}_2 = \mathbb{E}_{\boldsymbol{h} \sim \mathcal{H}_2}[f(\boldsymbol{h})]$, $\mathbf{\Sigma}_1 = \text{Cov}_{\boldsymbol{h} \sim \mathcal{H}_1}[f(\boldsymbol{h})]$, and $\mathbf{\Sigma}_2 = \text{Cov}_{\boldsymbol{h} \sim \mathcal{H}_2}[f(\boldsymbol{h})]$. Notably, we employ the encoder from CsiNet \cite{dlcsifeedback} as the function $f$ because CsiNet is specifically designed for CSI compression, making it more suitable for wireless communication tasks compared to the Inception network \cite{inception} that is trained on natural images.

Moreover, based on physical channel model assumptions, we derive a concise expression for the EMD (see Supplementary Text). This inherent domain knowledge provides a solid foundation for the hypernet to generalize across diverse channel distributions (Fig.~S4), requiring only prior statistics instead of raw channel samples. In the purely data-driven regime, where samples are required for environment identification, the consistency of the underlying physics still helps lower data requirements despite the absence of a closed-form expression (see Figs.~S5 and S6).

\subsubsection*{Hypernet Training Loss}
We train the hypernet so that its generated binary gate can represent the relationship between different environments. To this end, we train the hypernet using a cosine distance loss. Specifically, given input samples $\boldsymbol{h}_1 \sim \mathcal{H}_1$ and $\boldsymbol{h}_2 \sim \mathcal{H}_2$, and denoting the generated gates as $\boldsymbol{g}_1$ and $\boldsymbol{g}_2$, the loss function is defined as follows:
\begin{equation}
	\mathcal{L}_h(\boldsymbol{h}_1, \boldsymbol{h}_2) \triangleq \left(\frac{\boldsymbol{g}_1^\top \boldsymbol{g}_2}{\|\boldsymbol{g}_1\| \|\boldsymbol{g}_2\|} - \exp\bigl(-\beta \,\text{FCD}(\mathcal{H}_1, \mathcal{H}_2)\bigr)\right)^2 + \lambda_h\sum_{i=1,2}\bigl(\|\boldsymbol{g}_i\|_1 - \rho\bigr)^2,  
\end{equation}
where the exponential function maps the FCD to a similarity score in $[0,1]$, the last term controls the sparsity of the output gate, and $\beta$, $\lambda_h$, and $\rho$ are hyperparameters that control the sensitivity of the distance mapping, weight the sparsity regularization term, and define the target sparsity, respectively.

\subsection*{Spiking Rate Consolidation}
SRC is a regularization-based method tailored for SNNs that mitigates catastrophic forgetting with relatively low complexity. During backbone SNN training, we define the loss function as
\begin{equation}
	\mathcal{L}'(\boldsymbol{\theta}) = \mathcal{L}(\boldsymbol{\theta}) + \lambda \sum_{m} \sum_{k} \omega_{m,k} \big(\theta_k - \bar{\theta}_{m,k}\big)^2,
\end{equation}
where $\mathcal{L}'(\boldsymbol{\theta})$ and $\mathcal{L}(\boldsymbol{\theta})$ denote the total loss and the standard loss for the signal processing task, respectively. The vector $\boldsymbol{\theta}$ represents the current model parameters and $\theta_{k}$  denotes its $k$-th element. The scalar $\bar{\theta}_{m,k}$  denotes the $k$-th element of $\bar{\boldsymbol{\theta}}_{m}$, the parameter vector learned in a previously encountered environment $m$. The quantity $\omega_{m,k}$ denotes the importance of $\bar{\theta}_{m,k}$ and $\lambda$ denotes the regularization strength.

The core challenge for regularization-based methods lies in how to measure the importance $\omega_{m,k}$ of previously learned weights. Classic methods, such as EWC, rely on the FIM or its variants. Although effective, they incur high computational costs because they require backpropagation to compute gradients. Moreover, when using the common diagonal approximation, storing per-parameter importance values incurs a memory cost of $O(|\boldsymbol{\theta}|)$. To address this issue, we exploit the spike-based computing properties of SNNs and propose representing importance using the spiking rates of neurons. 

For clarity, we describe SRC for a single full connected (FC) layer in a given environment $m$. Let $N_{\text{pre}}$ and $N_{\text{post}}$ denote the numbers of pre-synaptic and post-synaptic neurons, respectively. Define $\boldsymbol{\mu}_m \in \mathbb{R}^{N_{\text{pre}}}$ and $\boldsymbol{\nu}_m \in \mathbb{R}^{N_{\text{post}}}$ as the average firing-rate vectors of the pre-synaptic and post-synaptic layers, with entries
\begin{equation}
	\mu_{m,i} \triangleq \frac{1}{T}\sum_{t=1}^{T} S^{\text{pre}}_{i}[t], \qquad 
	\nu_{m,j} \triangleq \frac{1}{T}\sum_{t=1}^{T} S^{\text{post}}_{j}[t],
\end{equation}
where $S^{\text{pre}}_{i}[t]\in\{0,1\}$ and $S^{\text{post}}_{j}[t]\in\{0,1\}$ denote the spikes of the $i$-th pre-synaptic and $j$-th post-synaptic neurons at time step $t$, respectively, and $T$ is the number of simulation time steps. The importance matrix is then defined as
\begin{equation}
	\boldsymbol{\Omega}_{m} \triangleq \boldsymbol{\nu}_{m}\boldsymbol{\mu}_{m}^\top \in \mathbb{R}^{N_{\text{post}} \times N_{\text{pre}}},
\end{equation}
and $\boldsymbol{\omega}_{m} \triangleq \mathrm{vec}(\boldsymbol{\Omega}_{m})$ denotes the vectorized form of $\boldsymbol{\Omega}_{m}$, with $\omega_{m,k}$ denoting the $k$-th element of $\boldsymbol{\omega}_{m}$. This design aligns with the principles of Hebbian learning, which strengthens synapses based on correlated pre- and post-synaptic activity. Here, spiking rates serve as a time-averaged measure of such activity. The computational complexity of calculating $\boldsymbol{\mu}_{m}$ and $\boldsymbol{\nu}_{m}$ is low since it only involves the accumulation of spikes. Moreover, SRC stores only the two firing-rate vectors and forms $\boldsymbol{\Omega}_{m}$ implicitly when needed, yielding a memory cost of $O(N_{\text{pre}} + N_{\text{post}})$, which is significantly lower than the $O(N_{\text{pre}} \times N_{\text{post}})$ cost of EWC or SI.

In our experiments, we determine the SRC regularization strength $\lambda$ via grid search. We also fine-tune the regularization strengths for EWC and SI. Table~\ref{tab:regularizer} summarizes the final values used for each evaluated task. The above introduction to SpikACom is presented using FC layers for clarity. Details on adapting the proposed SpikACom framework to convolutional architectures are provided in the Supplementary Text.

\begin{table}[h]
	\centering
	\caption{\textbf{Regularization strength $\lambda$ for different signal processing tasks}}
	\label{tab:regularizer}
	\begin{tabular}{lcccc}
		\hline
		Task & SRC & EWC & SI & SRC+Gate \\
		\hline
		Task-oriented communication & 1.5 & 0.5 & 0.5 & 0.2 \\

		MIMO beamforming            & 1.0 & 20 & 600 & 0.08 \\

		OFDM channel estimation      & 0.08 & 0.1 & 0.05 & 0.1 \\
		\hline
	\end{tabular}
\end{table}

\subsection*{Neuromorphic Semantic Communication}
\subsubsection*{System Model and Network Architecture}
We consider a multi-sensor semantic communication system with $I=2$ distributed sensors and evaluate our algorithms on the DVS128 Gesture dataset~\cite{dvsg_dataset}. Each sample is an event stream recorded by a dynamic vision sensor camera.
The gesture samples are uniformly divided into $I$ segments, each fed into the semantic encoder of the $i$-th sensor. 
Each semantic encoder consists of 5 stacked SNN blocks~\cite{snn_fang}. Each block includes a 2D convolutional layer (64 channels, $3\times3$ kernel, stride 1), batch normalization, a LIF neuron layer, and $2\times2$ max pooling. Subsequently, a channel encoder with two fully connected layers (2,400 and 64 spiking neurons, respectively) processes the features. The resulting spiking representations are transmitted over the physical channel using on-off keying (OOK) modulation with a block length of 32, to reflect the binary nature of spikes. At the receiver, a spiking decoder composed of a 1D convolutional layer (16 channels, kernel size 13), two fully connected layers (3{,}200 and 110 neurons), and a 10-dimensional voting layer recovers the final classification output~\cite{snn_fang}.

The system does not rely on explicit pilot symbols for equalization; instead, it is trained end-to-end based solely on channel statistics~\cite{semantic_yehao}. 
Moreover, the semantic encoders are trained under noise-free channels and remain fixed during subsequent tests since we focus on time-varying channels. 
Only the channel encoder and channel decoder are updated when the environments change.

\subsubsection*{Channel Model and Simulation Settings}
We augment the training data using random spatial translations, and the number of simulation time steps is set to $T = 16$. We adopt a multipath channel model with $L=8$ taps~\cite{semantic_yehao}. Let $x[n]$ denote the transmitted signal at discrete time index $n$. The received signal $y[n]$ is given by
\begin{equation}
	y[n] = \sum_{\ell=1}^{L} h_\ell \, x[n-\ell] + w[n],
\end{equation}
where $\boldsymbol{h} = [h_1, h_2, \ldots, h_{L}]^\top$ denotes the channel impulse response vector and $w[n] \sim \mathcal{CN}(0, \sigma^2)$ is additive white Gaussian noise (AWGN). Assuming that the transmitted signal is wide-sense stationary and independent of the channel, the signal-to-noise ratio (SNR) is defined as
\begin{equation}
	\begin{aligned}
		\mathrm{SNR} &\triangleq \frac{\mathbb{E}[|\sum_{\ell=1}^{L} h_\ell \, x[n-\ell]|^2]}{\mathbb{E}[|w[n]|^2]} \\
		&= \frac{\mathbb{E}[|x[n]|^2]}{\sigma^2} \cdot \sum_{\ell=1}^{L} \mathbb{E}[|h_\ell|^2].
	\end{aligned}
\end{equation}
We fix $\mathbb{E}[|x[n]|^2]/\sigma^2$ at $0$~dB and vary the channel gain $\sum_{\ell=1}^{L} \mathbb{E}[|h_\ell|^2]$ from $8$~dB to $-20$~dB in uniform steps of $4$~dB.

\subsection*{MIMO Beamforming}
We consider a downlink MIMO system in which a single base station (BS) serves $K$ users simultaneously. We focus on the following sum-rate maximization problem since the primary role of beamforming is to boost the system throughput:
\begin{subequations}\label{sum_rate_problem}
	\begin{align}
		\max_{\{\mathbf{V}_k\}} & \sum_{k=1}^{K} \alpha_k R_k, \\
		\text{s.t.} & \sum_{k=1}^{K} \text{Tr}(\mathbf{V}_k\mathbf{V}_k^H) \leq P_T, 
	\end{align} 
\end{subequations}
where $\mathbf{V}_k \in \mathbb{C}^{N_t\times d}$ is the beamforming matrix for user $k$, with $N_t$ and $d$ denoting the number of transmit antennas and data streams, respectively. For simplicity, we assume that all users require the same number of data streams $d$ and have the same number of receive antennas $N_r$. The parameters $\alpha_k$ and $P_T$ denote the weight of user $k$ and the total transmit power, respectively. The achievable rate of user $k$ is given by
\begin{equation}
	R_k = \log\det \left(\mathbf{I} + \mathbf{H}_k\mathbf{V}_k\mathbf{V}_k^H\mathbf{H}_k^H\left(\sum_{j \neq k}\mathbf{H}_k\mathbf{V}_j\mathbf{V}_j^H\mathbf{H}_k^H + \sigma_k^2\mathbf{I}\right)^{-1} \right), \label{rate}
\end{equation}
where $\mathbf{H}_k \in \mathbb{C}^{N_r\times N_t}$ represents the channel between the BS and user $k$, and $\sigma_k^2$ denotes the noise power at user $k$.

\subsubsection*{Weighted Minimum Mean-Squared Error (WMMSE) Algorithm} \label{sec:wmmse}
Problem \eqref{sum_rate_problem} is challenging to solve due to the highly coupled non-convex rate expression \eqref{rate}. To tackle this issue, the WMMSE~\cite{WMMSE} algorithm equivalently reformulates problem \eqref{sum_rate_problem} into the following more tractable form:
\begin{subequations}
	\begin{align}
		\min_{\{\mathbf{W}_k,\mathbf{U}_k,\mathbf{V}_k\}} & \sum_{k=1}^{K}\alpha_k\bigl(\text{Tr}(\mathbf{W}_k\mathbf{E}_k) - \log \det(\mathbf{W}_k)\bigr), \\
		\text{s.t.} & \sum_{k=1}^{K}\text{Tr}(\mathbf{V}_k\mathbf{V}_k^H) \leq P_T,		
	\end{align}\label{problem_wmmse}
\end{subequations}
where $\mathbf{W}_k \in \mathbb{C}^{d\times d}$ is an auxiliary weight matrix for user $k$ and 
\begin{equation}
	\mathbf{E}_k \triangleq (\mathbf{I}-\mathbf{U}_k^H\mathbf{H}_k\mathbf{V}_k)(\mathbf{I}-\mathbf{U}_k^H\mathbf{H}_k\mathbf{V}_k)^H + \sum_{j \neq k}\mathbf{U}_k^{H}\mathbf{H}_k\mathbf{V}_j\mathbf{V}_j^H\mathbf{H}_k^H\mathbf{U}_k + \sigma_k^2\mathbf{U}_k^H\mathbf{U}_k,
\end{equation}
with $\mathbf{U}_k \in \mathbb{C}^{N_r \times d}$ being a virtual receive beamforming matrix for user $k$. 

The WMMSE algorithm uses block coordinate descent to solve problem \eqref{problem_wmmse}. Specifically, the variables $\mathbf{W}_k$, $\mathbf{U}_k$, and $\mathbf{V}_k$ are optimized in turn, with their optimal solutions given by
\begin{align}
	\mathbf{U}_k &= \mathbf{A}_k^{-1}\mathbf{H}_k\mathbf{V}_k, \label{U_update} \\ 
	\mathbf{W}_k &= \mathbf{E}_k^{-1},  \label{W_update}\\ 
	\mathbf{V}_k &= \alpha_k\mathbf{B}^{-1}\mathbf{H}_k^H\mathbf{U}_k\mathbf{W}_k, \label{V_update} 
\end{align}
respectively, where
\begin{equation}
	\mathbf{A}_k = \frac{\sigma_k^2}{P_T}\sum_{j=1}^{K}\text{Tr}(\mathbf{V}_j\mathbf{V}_j^H)\mathbf{I} + \sum_{j=1}^{K} \mathbf{H}_k\mathbf{V}_j\mathbf{V}_j^H\mathbf{H}_k^H,
\end{equation}
\begin{equation}
	\mathbf{E}_k = \mathbf{I} - \mathbf{U}_k^H\mathbf{H}_k\mathbf{V}_k,
\end{equation}
and
\begin{equation}
	\mathbf{B} = \sum_{j=1}^{K} \frac{\sigma_j^2}{P_T}\text{Tr}(\alpha_j\mathbf{U}_j\mathbf{W}_j\mathbf{U}_j^H)\mathbf{I} + \sum_{j=1}^{K}\alpha_j\mathbf{H}_j^H\mathbf{U}_j\mathbf{W}_j\mathbf{U}_j^H\mathbf{H}_j.
\end{equation}

\subsubsection*{Domain Knowledge Embedded Beamforming Layer}
The goal of beamforming is to design the beamforming matrices $\mathbf{V}_k$. However, directly setting the network output to $\mathbf{V}_k$ results in poor performance because such a purely black-box structure cannot effectively handle the sophisticated beamforming optimization. To address this, we draw inspiration from the WMMSE algorithm and incorporate domain knowledge into the design of SNN-based beamforming. By reviewing the WMMSE formulation, we observe that a core aspect of the algorithm is the introduction of auxiliary variables $\mathbf{W}_k$ and $\mathbf{U}_k$. Therefore, we let the network output be $\mathbf{W}_k$ and $\mathbf{U}_k$, and compute $\mathbf{V}_k$ by substituting them into Eq.~\eqref{V_update}. This approach offers two main benefits. First, the dimensions of $\mathbf{W}_k$ and $\mathbf{U}_k$ are significantly smaller than those of $\mathbf{V}_k$, which reduces the network size and the search space. Second, Eqs.~\eqref{U_update}--\eqref{V_update} represent necessary conditions for the optimal solution to problem \eqref{problem_wmmse}~\cite{WMMSE}, meaning that any optimal $\mathbf{V}_k$ must also satisfy these equations. Therefore, incorporating Eq.~\eqref{V_update} into the network specifies the solution structure and embeds the necessary domain knowledge for beamforming~\cite{Bf_unfolding1,Bf_unfolding2,Bf_unfolding3,Bf_unfolding4}. Denoting the network outputs as $\bar{\mathbf{W}}_k$ and $\bar{\mathbf{U}}_k$, the model-driven layer is formally expressed as
\begin{equation}
	\mathbf{V}_k = \alpha_k\mathbf{B}^{-1}\mathbf{H}_k^H\bar{\mathbf{U}}_k\bar{\mathbf{W}}_k,
\end{equation}
where
\begin{equation}
	\mathbf{B} = \sum_{j=1}^{K} \frac{\sigma_j^2}{P_T}\text{Tr}(\alpha_j\bar{\mathbf{U}}_j\bar{\mathbf{W}}_j\bar{\mathbf{U}}_j^H)\mathbf{I} + \sum_{j=1}^{K}\alpha_j\mathbf{H}_j^H\bar{\mathbf{U}}_j\bar{\mathbf{W}}_j\bar{\mathbf{U}}_j^H\mathbf{H}_j.
\end{equation}

\subsubsection*{Simulation Settings and Network Architecture}
We assume $K=6$, $N_t = 64$, and $N_r = d = 4$. Moreover, we adopt the DeepMIMO channel, where the system configuration is provided in Table~S1. To generate different wireless environments, we divide the selected region into $4\, \text{m}\times 4\, \text{m}$ uniform blocks. An environment is specified by a $K$-dimension index vector, with the $k$th element representing the block in which user $k$ is located. 

The backbone network, referred to as SpikingBFNet, consists of a 2D convolutional layer (36 channels, $3\times3$ kernel, stride 1), a hidden FC layer ($3,000$ spiking neurons), and an output FC layer. The output is averaged over $T=4$ simulation time steps and reshaped into $\bar{\mathbf{W}}_k$ and $\bar{\mathbf{U}}_k$, which are then fed into the model-driven layer. 

\subsection*{Channel Estimation}
We consider a single-input single-output (SISO) communication link. In OFDM systems, the input-output relationship at the $i$-th subcarrier and the $k$-th OFDM symbol is described as:
\begin{equation}
	Y[i,k] = H[i,k] X[i,k] + Z[i,k],
\end{equation}
where $H[i,k]$ denotes the $(i,k)$-th element of the overall frequency-time channel response matrix $\mathbf{H} \in \mathbb{C}^{N_c \times N_s}$, with $N_c$ and $N_s$ denoting the number of subcarriers and OFDM symbols, respectively. $X[i,k]$, $Y[i,k]$, and $Z[i,k]$ denote the transmitted signal, the received signal, and the AWGN, respectively.

\subsubsection*{Least Squares (LS) Estimator}
The goal of channel estimation is to recover $\mathbf{H}$ based on transmitted pilots $X[p]$ and the received pilots $Y[p]$. Due to limited spectral resources, pilots are sparsely inserted into the frequency-time grid. To formalize this, we define the set of all pilot indices as $\mathcal{P} \subset \mathcal{M}$, where $\mathcal{M} \triangleq \{(i,k) \mid 1 \leq i \leq N_c, 1 \leq k \leq N_s\}$ represents the indices of the entire frequency-time grid. The LS estimator at the pilot positions is given by:
\begin{equation}
	\hat{H}[p] = \frac{Y[p]}{X[p]}, \quad \forall p \in \mathcal{P}.
\end{equation}
The LS estimator minimizes the squared error between the received and transmitted pilot symbols. The channel coefficients at the remaining data positions, $\{\hat{H}[q] \mid q \in \mathcal{M} \setminus \mathcal{P} \}$, are typically obtained via interpolation. In this work, we employ radial basis function (RBF) interpolation~\cite{rbf} as a baseline as it generally provides superior accuracy compared to linear or cubic interpolation schemes.

\subsubsection*{Linear Minimum Mean-Squared Error (LMMSE) Estimator}
The LMMSE estimator offers improved performance over the LS estimator by exploiting the second-order channel statistics and noise power information. Specifically, the LMMSE estimator applies a linear filter $\mathbf{A}_{\text{LMMSE}}$ to the LS estimates at pilot positions, denoted as $\hat{\boldsymbol{h}}_p \in \mathbb{C}^{|\mathcal{P}|}$:
\begin{equation}
	\hat{\boldsymbol{h}}_{\text{LMMSE}} = \mathbf{A}_{\text{LMMSE}} \hat{\boldsymbol{h}}_p.
\end{equation}
The optimal LMMSE estimator minimizes the MSE between the vectorized true channel $\boldsymbol{h}$ and the estimated channel:
\begin{equation}
	\min_{\mathbf{A}_{\text{LMMSE}}} \; \text{MSE} \triangleq \mathbb{E}[\|\boldsymbol{h} - \mathbf{A}_{\text{LMMSE}}\hat{\boldsymbol{h}}_p\|_2^2].
\end{equation}
The optimal matrix $\mathbf{A}_{\text{LMMSE}}$ is derived as:
\begin{equation}
	\mathbf{A}_{\text{LMMSE}} = \mathbf{R}_{\boldsymbol{h}\boldsymbol{h}_p}(\mathbf{R}_{\boldsymbol{h}_p\boldsymbol{h}_p} + \sigma^2\mathbf{I})^{-1},
\end{equation}
where $\mathbf{R}_{\boldsymbol{h}\boldsymbol{h}_p} \triangleq \mathbb{E}[\boldsymbol{h}\boldsymbol{h}_p^H]$ denotes the cross-correlation matrix between the channel response of the entire grid and the channel response at pilot positions, with $\boldsymbol{h}_p$ representing the vector of true channel coefficients at pilot positions. Similarly, $\mathbf{R}_{\boldsymbol{h}_p\boldsymbol{h}_p} \triangleq \mathbb{E}[\boldsymbol{h}_p\boldsymbol{h}_p^H]$ denotes the autocorrelation matrix of $\boldsymbol{h}_p$, and $\sigma^2$ represents the noise variance. In our experiments, these statistics are estimated empirically from the training dataset.

\subsubsection*{Simulation Settings and Network Architecture}
We define an OFDM frequency-time grid with $N_c = 72$ subcarriers and $N_s = 14$ OFDM symbols. The number of pilots is set to $N_p = 48$, and pilots are inserted following the pattern described in~\cite{OFDM_CE1}. We adopt the WINNER II channel model, with detailed configurations provided in Table~S2. To ensure diversity, we simulate eight distinct wireless propagation scenarios: A1 (Indoor small office)-LOS, A2 (Indoor to outdoor)-NLOS, B2 (Bad urban micro-cell)-NLOS, C1 (Urban macro-cell)-NLOS, C2 (Urban macro-cell)-NLOS, C3 (Bad urban macro-cell)-NLOS, C4 (Outdoor to indoor macro-cell)-NLOS, and D1 (Rural macro-cell)-LOS. For each scenario, we generate $9,000$ samples for training and $1,000$ samples for testing.

For the backbone network, we adopt an architecture inspired by ReEsNet~\cite{OFDM_CE2} with two key improvements. First, the transposed convolution layer used in ReEsNet for upsampling is replaced with an FC layer to enhance the representation capacity. Second, we introduce a residual connection $\boldsymbol{y}_{\text{res}} = \mathbf{A}_{\text{res}}\hat{\boldsymbol{h}}_p$ that directly links the input LS estimates to the final output. This residual path compensates for potential information loss during the spike conversion process, acting as a learned approximation of the LMMSE estimator. Moreover, to mitigate the computational overhead associated with the matrix multiplication of $\mathbf{A}_{\text{res}}$, we employ a low-rank factorization:
\begin{equation}
	\mathbf{A}_{\text{res}} = \mathbf{A}_L \mathbf{A}_R,
\end{equation}
where $\mathbf{A}_L \in \mathbb{C}^{N_{\text{tot}} \times r}$ and $\mathbf{A}_R \in \mathbb{C}^{r \times N_p}$, with $N_{\text{tot}} = N_c N_s$ denoting the total grid size and $r < N_{\text{tot}}$ being the rank. In our experiments, we set the reduced rank to $r = \frac{1}{2} N_{\text{tot}}$.

\subsection*{Energy Consumption Modeling}
To ensure a unified evaluation across different algorithms, we estimate energy consumption by counting operations. Each
operation is mapped to physical energy consumption based on a 45~nm CMOS
process~\cite{snnleadfed}, where a MAC
operation consumes $E_{\text{mac}}=3.2~\mathrm{pJ}$ and an AC operation
consumes $E_{\text{ac}}=0.1~\mathrm{pJ}$~\cite{energy_model}. The total energy is given by
\begin{equation}
	E_{\text{total}} = E_{\text{mac}} N_{\text{mac}} + E_{\text{ac}} N_{\text{ac}},
\end{equation}
where $N_{\text{mac}}$ and $N_{\text{ac}}$ denote the number of MAC and AC operations, respectively. A comprehensive derivation of the operation counts for each algorithm is provided in the Supplementary Text.

\newpage
\section*{Supplementary Materials}
\subsection*{Generalization Bounds under EMD}\label{sec1}
Continual learning involves sequentially updating a model under a stream of environments with limited access to past data, which can cause catastrophic forgetting when the current update degrades performance on previously encountered environments. In this setting, a central question is: what are the theoretical limits of retaining past knowledge while adapting to new data? To address this, we analyze domain generalization bounds, which quantify the performance gap when a model trained on one distribution is evaluated on another. These bounds formalize how distribution shift, measured by the earth mover's distance (EMD), governs the inherent trade-off between adaptation (learning new environments) and retention (preserving past knowledge).

Consider two environments characterized by distributions $Q$ (source) and $P$ (target). A general form of the domain generalization bound can be expressed as follows:
\begin{equation}
	R_P(h,f) \leq R_Q(h,f) + \Delta(Q,P),
\end{equation}
where $R_P(h,f) \triangleq \mathbb{E}_{x \sim P}[L(h(x), f(x))]$ is the risk on target domain $P$, $h$ is the hypothesis, $f$ is the true labelling function, and $L$ is the loss function. Likewise, $R_Q(h,f)$ denotes the risk on source domain $Q$. The term $\Delta(Q,P)$ represents a divergence measure between $Q$ and $P$. In this work, we employ EMD as the distance measure. For distributions $Q$ and $P$, the EMD is defined by \footnote{We define EMD here as the 1-Wasserstein distance ($W_1$) because the Kantorovich–Rubinstein duality applies directly to $W_1$, which is slightly different from the main text, where the closed-form FCD for two Gaussians is derived under the 2-Wasserstein distance ($W_2$). However, for distributions with finite second moments, the relationship $W_1(Q,P) \leq W_2(Q,P)$ holds; therefore, the results derived here for $W_1$ can be easily extended to the $W_2$ case.
}
\begin{equation}
	\text{EMD}(Q, P) = \inf_{\gamma \in \Gamma(Q, P)}\int_{\mathcal{X} \times \mathcal{X}}||x-y||d\gamma(x,y),
\end{equation}
where $\mathcal{X}$ denotes the sample space and $\Gamma(Q, P)$ denotes the set of all joint distributions whose marginals are $Q$ and $P$. 
Under EMD, the domain generalization bound can be written explicitly as follows:
\begin{lemma}
	Assume the loss function $L$ is $M$-Lipschitz continuous and that $h$ and $f$ are $B_h$- and $B_f$-Lipschitz continuous, respectively. Then, the following bound holds:
	\begin{equation}
		R_P(h,f) - R_Q(h,f) \leq G\cdot{\rm EMD}(Q,P),
	\end{equation}  
	where $G = 2M\cdot\max(B_h,B_f)$.
\end{lemma}

\begin{proof}
	Define \(g(x) \triangleq L\bigl(h(x), f(x)\bigr)\). Then we have
	\begin{subequations}
		\begin{align}
			\bigl|g(x_1) - g(x_2)\bigr|
			&= \bigl|L\bigl(h(x_1), f(x_1)\bigr) - L\bigl(h(x_2), f(x_2)\bigr)\bigr| \\
			&\leq M \sqrt{\|h(x_1) - h(x_2)\|^2 + \|f(x_1) - f(x_2)\|^2} \\
			&\leq M \Bigl(\|h(x_1) - h(x_2)\| + \|f(x_1) - f(x_2)\|\Bigr) \\
			&\leq MB_h \|x_1 - x_2\| + MB_f \|x_1 - x_2\| \\
			&\leq G \|x_1 - x_2\|,
		\end{align}
	\end{subequations}
	where the first inequality applies the Lipschitz continuity of \( L \), the second inequality applies the triangle inequality, and the third inequality applies the Lipschitz continuity of \( h \) and \( f \). Hence, \( g(x) \) is \( G \)-Lipschitz continuous. Next,
	\begin{subequations}
		\begin{align}
			R_P(h,f) - R_Q(h,f)
			&= \mathbb{E}_{x \sim P}[L\bigl(h(x), f(x)\bigr)] 
			- \mathbb{E}_{x \sim Q}[L\bigl(h(x), f(x)\bigr)] \\
			&= \mathbb{E}_{x \sim P}[g(x)] - \mathbb{E}_{x \sim Q}[g(x)] \\
			&\leq \sup_{\|\phi\|_L \le G} 
			\Bigl(\mathbb{E}_{x \sim P}[\phi(x)] - \mathbb{E}_{x \sim Q}[\phi(x)]\Bigr) \\
			&= G \cdot \text{EMD}(Q,P),
		\end{align}
	\end{subequations}
	where the last equality uses the Kantorovich–Rubinstein duality (\textit{58}) %\cite{WGAN}. 
	This completes the proof.
\end{proof}

Lemma 1 shows that the change in expected loss between two environments is controlled by their distributional distance, implying that sequential updates are less likely to induce forgetting when consecutive environments are close in EMD. Then, the following theorem quantifies how much $R_P(h,f)$ may deviate from the performance achieved by training directly on the target distribution $P$.
\begin{theorem}
	Let \(h_Q^* \in \arg\min_{h \in \mathcal{H}} R_Q(h,f)\) and 
	\(h_P^* \in \arg\min_{h \in \mathcal{H}} R_P(h,f)\). Assume the assumptions of Lemma 1 hold and that \(L\) is symmetric and satisfies the triangle inequality. Then we have
	\begin{equation}
		R_P(h,f) \leq R_P(h_P^*,f) + R_Q(h, h_Q^*) + G\cdot\text{EMD}(Q,P) + R_P(h_Q^*,h_P^*).
	\end{equation}
\end{theorem}

\begin{proof}
	By the triangle inequality and Lemma 1, we obtain
	\begin{subequations}
		\begin{align}
			R_P(h,f) &\leq R_P(h,h_Q^*) + R_P(h_Q^*,h_P^*) + R_P(h_P^*,f), \\
			& \leq R_Q(h,h_Q^*) + G\cdot\text{EMD}(Q,P) + R_P(h_Q^*,h_P^*) + R_P(h_P^*,f),
		\end{align}
	\end{subequations}
	where the first inequality follows from the triangle inequality and the second follows from Lemma 1. This completes the proof.
\end{proof}

The above theorem is a direct application of Theorem 8 in (\textit{67}),%\cite{generalization_bound}, 
where the discrepancy distance is replaced by the EMD. Note that $R_Q(h,h_Q^*)$ is typically small because the hypothesis $h$ is trained on source distribution $Q$, resulting in  performance close to that of optimal hypothesis $h_Q^*$. Moreover, it is reasonable to assume that the average loss between the best in-class hypotheses, $L_P(h_Q^*,h_P^*)$, correlates with the distributional distance $\text{EMD}(Q,P)$ (\textit{67}).%\cite{generalization_bound},
Consequently, Theorem 1 indicates that the generalization error is primarily bounded by $\text{EMD}(Q,P)$. When two consecutive environments are sufficiently similar (small $\mathrm{EMD}(Q,P)$), a model adapted to the current environment can retain performance on previously encountered ones with limited additional regularization. In contrast, large $\mathrm{EMD}(Q,P)$ implies a larger unavoidable discrepancy, motivating mechanisms that allocate new plasticity or separate parameters to prevent interference. 

Recall that our SpikACom framework can intelligently control the spiking neurons of the backbone network based on distributional distances. For two similar distributions, the hypernet will activate a similar subset of spiking neurons, effectively facilitating knowledge sharing. Therefore, our algorithm is in line with the above theoretical result, which reduces the need for extensive retraining across similar domains and significantly improves efficiency.  

The following theorem provides a lower bound on the generalization error but relies on more stringent assumptions.
\begin{assumption}
	Hypothesis \(h_Q\) trained on the source domain satisfies
	\begin{equation}
		-\frac{1}{K}\ln q(x) -\delta 
		\;\leq\; L\bigl(h_Q(x),f(x)\bigr) - L\bigl(h_Q^*(x),f(x)\bigr) 
		\;\leq\; -\frac{1}{K}\ln q(x) +\delta,
	\end{equation}
	where \(K\) and \(\delta\) are constants.
\end{assumption}

\begin{assumption}
	\begin{equation}
		\text{EMD}(Q,P) \;\leq\; \sqrt{\frac{2\,\mathrm{KL}(P \,\|\, Q)}{\rho}},
	\end{equation}
	where \(\rho\) is a constant.
\end{assumption}

\begin{remark}
	Assumption 1 reflects a natural phenomenon: when a hypothesis \(h_Q\) is trained on samples drawn from source distribution \(Q\), it should be close to the optimal hypothesis \(h_Q^*\). Furthermore, for regions of higher density q(x), there is a greater likelihood of sampling points in those areas during training, which improves performance locally. Hence, we assume that \(L\bigl(h_Q(x),f(x)\bigr) - L\bigl(h_Q^*(x),f(x)\bigr)\) is proportional to the data density \(-\ln q(x)\) and their absolute difference is bounded by $\delta$. 
	
	Assumption 2 states that the EMD is bounded above by the square root of the KL divergence scaled by some constant. In general, the EMD is expected to be positively correlated with the KL divergence since both measure distance between distributions. However, since the KL divergence is not a true distance, certain extreme cases (e.g., two sharply peaked Gaussian distributions with different means) may invalidate this assumption. We stress that such cases rarely arise in practice. Hence, to streamline our analysis, we adopt this inequality as an assumption. Notably, the same inequality also arises as a sufficient condition for the Talagrand inequality (\textit{68}).%\cite{talagrand}.
\end{remark}
\begin{lemma}
	Assuming that Assumption 1, Assumption 2, and the assumptions of Lemma 1 hold, we have
	\begin{align}
		R_P(h_Q,f) - R_Q(h_Q,f) 
		&\geq R_P(h_Q^*,f) - R_Q(h_Q^*,f) - 2\delta 
		+ \frac{1}{K}\bigl(H(P) - H(Q)\bigr) \\
		&\quad + \frac{\rho}{2K}\,\text{EMD}(Q,P)^2. \nonumber
	\end{align}	
\end{lemma}

\begin{proof}
	\begin{subequations}
		\begin{align}
			R_P&(h_Q,f) - R_Q(h_Q,f) \\
			&= \int L\bigl(h_Q(x), f(x)\bigr)\,p(x)\,dx - \int L\bigl(h_Q(x), f(x)\bigr)\,q(x)\,dx \\[6pt]
			&\geq \int \bigl[L\bigl(h_Q^*(x), f(x)\bigr) - \frac{\ln q(x)}{K} - \delta\bigr]\,p(x)\,dx 
			- \int \bigl[L\bigl(h_Q^*(x), f(x)\bigr) - \frac{\ln q(x)}{K} + \delta\bigr]\,q(x)\,dx \label{assumption1}\\[6pt]
			&= R_P(h_Q^*,f) - \frac{1}{K}\int \ln q(x)\,p(x)\,dx - 2\delta - R_Q(h_Q^*,f) 
			- \frac{1}{K}H(Q) \\[6pt]
			&= R_P(h_Q^*,f) - R_Q(h_Q^*,f) - 2\delta 
			+ \frac{1}{K}\Bigl(H(P) - H(Q) + \mathrm{KL}(P\|Q)\Bigr)\label{KL_divergence} \\[6pt] 
			&\geq R_P(h_Q^*,f) - R_Q(h_Q^*,f) - 2\delta 
			+ \frac{1}{K}\bigl(H(P) - H(Q)\bigr) 
			+ \frac{\rho}{2K}\,\text{EMD}(Q,P)^2, 
		\end{align}
	\end{subequations}
	where \eqref{assumption1} applies Assumption 1, and \eqref{KL_divergence} follows from the definition of the KL divergence. The final inequality applies Assumption 2. 
\end{proof}

The following theorem states a lower bound on how much \(R_P(h_Q,f)\) can deviate from \(R_P(h_P^*,f)\).

\begin{theorem}
	Assuming that the conditions of Lemma 2 hold, we have the following lower bound on the generalization error:
	\begin{align}
		R_P(h_Q,f) - R_P(h_P^*,f) 
		&\geq \frac{\rho}{2K}\,\text{EMD}(Q,P)^2 
		+ \frac{1}{K}\bigl(H(P) - H(Q)\bigr) 
		- 2\delta.
	\end{align}
\end{theorem}

\begin{proof}
	\begin{subequations}
		\begin{align}
			R_P&(h_Q,f)-R_P(h_P^*,f)  = R_P(h_Q,f)- R_Q(h_Q,f) + R_Q(h_Q,f)- R_P(h_P^*,f) \\
			&
			\geq R_P(h_Q^*,f)-R_Q(h_Q^*,f)-2\delta +\frac{1}{K}\big(H(P) - H(Q)\big) + \frac{\rho}{2K}{\rm EMD}(Q,P)^2 \nonumber\\
			&+ R_Q(h_Q,f) - R_P(h_P^*,f) \label{lemma2} \\
			& = R_P(h_Q^*,f)- R_P(h_P^*,f) + R_Q(h_Q,f) - R_Q(h_Q^*,f) 
			-2\delta +\frac{1}{K}\big(H(P) - H(Q)\big) \nonumber\\
			&+ \frac{\rho}{2K}{\rm EMD}(Q,P)^2  \\
			& \geq 
			\frac{\rho}{2K}{\rm EMD}(Q,P)^2+\frac{1}{K}\big(H(P) - H(Q)\big)-2\delta,
		\end{align}	
	\end{subequations}
	where \eqref{lemma2} applies Lemma 2 and the last inequality follows from $R_P(h_Q^*,f)- R_P(h_P^*,f) \geq 0$ and $R_Q(h_Q,f) - R_Q(h_Q^*,f) \geq 0$. This completes the proof. 
\end{proof}

Theorem 2 establishes that the generalization error is lower bounded by three terms. The first term is the square of the EMD. The second term captures the difference between the target distribution’s entropy and the source domain’s entropy. This suggests that training on a more diverse data distribution can enhance generalization ability. However, in our continual learning setup, the entropies of different wireless environments are typically similar, rendering this term negligible. The third term measures  how strictly Assumption 1 holds, which we consider to be small. Therefore, the principal factor contributing to the lower generalization bound is the EMD term.

Since the EMD is symmetric, two distributions with a large EMD will yield high generalization errors in both directions. This observation implies that, when only data from the current environment is available, a single hypothesis $h$ cannot concurrently achieve optimal performance on both distributions if they differ significantly. Regularization-based methods, such as EWC, fall into this category and thus perform poorly when the distance between distributions is large. Theorem 2 therefore underscores the necessity of introducing new plasticity in continual learning, especially when the new distribution differs significantly from all previously learned ones. The developed SpikACom embodies this principle by activating a different subset of spiking neurons upon encountering a distinctly new environment. Hence, it can simultaneously achieve satisfactory performance on both current and previous environments.    

%We validate the above theorems using experimental results. Fig.~\ref{fig:EMD_gen_error} compares the EMD and the generalization error for the OFDM channel estimation task, where the EMD is computed via our proposed FCD approach for reduced computational cost. As shown, the EMD is positively related to the generalization error.

%\begin{figure}[t]
%	\centering
%	\begin{subfigure}[b]{0.49\textwidth}
	%		\centering
	%		\includegraphics[width=\textwidth]{OFDMCE_FCD}
	%		\caption{}
	%		\label{fig:4a}
	%	\end{subfigure}
%	\hfill % Creates horizontal space between the figures
%	% Subfigure 2 (Panel b)
%	\begin{subfigure}[b]{0.49\textwidth}
	%		\centering
	%		\includegraphics[width=\textwidth]{OFDMCE_generalization_error.eps}
	%		\caption{}
	%		\label{fig:4b}
	%	\end{subfigure}
%	\caption{Comparison between the EMD and the generalization error $R_P(h,f)-R_P(h_P^*,f)$ in OFDM channel estimation task. (a) The EMDs between the eight WINNERII channel model scenarios. (b) The generalization errors between the eight channel distributions, which are normalized for better visualization.}
%	\label{fig:EMD_gen_error}
%\end{figure}

\newpage	
\subsection*{Derivation of Model-Based EMD}
In this section, we take a complementary perspective and derive an EMD measure directly from the channel model. 
Unlike data-driven approaches, this method does not rely solely on collected samples; 
instead, it also exploits the physical channel model, which remains invariant across environments. 
As a result, the derived EMD requires significantly fewer training samples and allows the hypernet to generalize more easily across different channel conditions, ensuring robustness and adaptability.

We consider a general spike transmission scenario aligned with our neuromorphic semantic communication task. 
Let the transmitted spike vector within one time slot be
\begin{equation}
	\boldsymbol{x} = [x[1],\ldots,x[d]]^\top \in \{0,1\}^d.
\end{equation}
Since our focus lies on channel distribution shifts, we assume for simplicity that
\begin{equation}
	x[n] \overset{\text{i.i.d.}}{\sim} \mathrm{Bernoulli}(\rho), 
	\quad n = 1,...,d,
\end{equation}
where $\rho$ denotes the average firing rate of spiking neurons.
The wireless channel is an $L$-tap fast-fading multipath channel, where each tap
\begin{equation}
	h_\ell[n] \sim \mathcal{CN}(0,p_\ell),
\end{equation}
is independently resampled at each time index. 
The total channel power is
\begin{equation}
	P_{\boldsymbol{h}} = \sum_{\ell=1}^{L} p_\ell.
\end{equation}
The received sample at time index $n$ is
\begin{equation}
	y[n] = \sum_{\ell=1}^{L} h_\ell[n]\,x[n-\ell] + w[n],
\end{equation}
where $w[n]\sim\mathcal{CN}(0,N_0)$ denotes the additive noise.

It can be readily checked that
\begin{equation}
	\operatorname{Var}(y[n]) = \rho P_{\boldsymbol{h}} + N_0,
\end{equation}
and 
\begin{equation}
	\operatorname{Cov}(y[n],y[m]) = 0, \qquad n \neq m.
\end{equation}
Thus,
\begin{equation}
	\operatorname{Cov}(\boldsymbol{y})
	= (\rho P_{\boldsymbol{h}} + N_0)\,\mathbf{I}_d.
\end{equation}

Assuming that the number of paths $L$ is large, the Lindeberg--Feller 
central limit theorem (\textit{66})%\cite{clt} 
gives
\begin{equation}
	\sum_{\ell=1}^{L} h_\ell[n]\,x[n-\ell]
	\approx \mathcal{CN}(0,\rho P_{\boldsymbol{h}}).
\end{equation}
Adding AWGN yields
\begin{equation}
	y[n] \approx \mathcal{CN}(0,\rho P_{\boldsymbol{h}} + N_0).
\end{equation}
Since the entries are uncorrelated and each satisfies a scalar CLT, the 
Cramér--Wold theorem (\textit{66})%\cite{clt} 
implies
\begin{equation}
	\boldsymbol{y} \approx 
	\mathcal{CN}\bigl(\boldsymbol{0},(\rho P_{\boldsymbol{h}} + N_0)\mathbf{I}_d\bigr).
\end{equation}

For two zero-mean Gaussians with covariances 
$\sigma_1^2 \mathbf{I}_d$ and $\sigma_2^2 \mathbf{I}_d$, 
the EMD is
\begin{equation}
	\text{EMD} = \sqrt{d}\,|\sigma_1 - \sigma_2|.
\end{equation}
Hence, the EMD between the received distributions 
under two channel distributions is
\begin{equation}
	\text{EMD}(\mathcal{H}_1,\mathcal{H}_2)
	\approx 
	\sqrt{d}\,
	\bigl|
	\sqrt{\rho P_{\boldsymbol{h}_1}+N_0}
	-
	\sqrt{\rho P_{\boldsymbol{h}_2}+N_0}
	\bigr|.
\end{equation}
Defining $\mathrm{SNR}_i = \rho P_{\boldsymbol{h}_i}/N_0$, this becomes
\begin{equation}
	\text{EMD}(\mathcal{H}_1,\mathcal{H}_2)
	\approx \sqrt{dN_0}\,
	\bigl|
	\sqrt{1+\mathrm{SNR}_1}
	-
	\sqrt{1+\mathrm{SNR}_2}
	\bigr|.
\end{equation}

Since only the relative distance matters in our setting, the constant 
$\sqrt{dN_0}$ can be ignored. 
Thus, under a fast-fading channel with rich scattering, the detailed multipath 
structure is averaged out, and the EMD between environments is determined 
primarily by the difference in their SNR values. This model-based EMD facilitates learning the relationships between different channels, as illustrated in Fig.~\ref{fig:EMD_vs_gate_semantic}.

\newpage	
\subsection*{Extension of SpikACom to Convolutional Layers}
The preceding sections in the main text detailed the SpikACom framework within the context of fully connected (FC) layers. In this section, we outline the adaptation of the hypernet-based context modulation and SRC mechanisms for spiking convolutional layers.

\subsubsection*{Hypernet-based context modulation for convolutional layers}
To start with, we analyze the FC layer, showing that orthogonal binary gates lead to no-interference during learning.
Consider a fully connected layer with weight matrix
$\mathbf{W}\in\mathbb{R}^{N_{\text{post}}\times N_{\text{pre}}}$,
input spike vector $\boldsymbol{s}\in\mathbb{R}^{N_{\text{pre}}}$,
and a hypernet-generated binary gate
$\boldsymbol{g}\in\{0,1\}^{N_{\text{post}}}$.
The gated output is
\begin{equation}
	\boldsymbol{y}=\boldsymbol{g}\odot f(\mathbf{W}\boldsymbol{s}),
\end{equation}
where $f(\cdot)$ denotes the component-wise neuron nonlinearity.

For two environments $e\in\{1,2\}$, we denote the loss, input spike vector, and gate as $L_e$, $\boldsymbol{s}_e$, and $\boldsymbol{g}_e$, respectively.
By the chain rule, the gradient with respect to $\mathbf{W}$ admits an outer-product form
\begin{equation}
	\nabla_{\mathbf{W}}L_e=\boldsymbol{\delta}_e\,\boldsymbol{s}_e^\top,
\end{equation}
where
\begin{equation}
	\boldsymbol{\delta}_e
	\triangleq
	\bigl(\nabla_{\boldsymbol{y}_e}L_e\bigr)\odot \boldsymbol{g}_e \odot f'(\mathbf{W}\boldsymbol{s}_e)
\end{equation}
and $\odot$ denotes the Hadamard product.

To quantify interference between two environments, consider the Frobenius inner product
\begin{equation}
	\left\langle \nabla_{\mathbf{W}}L_1,\ \nabla_{\mathbf{W}}L_2 \right\rangle_F
	=
	\left\langle \boldsymbol{\delta}_1\boldsymbol{s}_1^\top,\ \boldsymbol{\delta}_2\boldsymbol{s}_2^\top \right\rangle_F
	=
	(\boldsymbol{\delta}_1^\top\boldsymbol{\delta}_2)\,(\boldsymbol{s}_1^\top\boldsymbol{s}_2).
\end{equation}
For disjoint binary gates, $\boldsymbol{g}_1\odot \boldsymbol{g}_2=\mathbf{0}$,
the supports of $\boldsymbol{\delta}_1$ and $\boldsymbol{\delta}_2$ do not overlap, yielding
$\boldsymbol{\delta}_1^\top\boldsymbol{\delta}_2=0$ and hence
\begin{equation}
	\left\langle \nabla_{\mathbf{W}}L_1,\ \nabla_{\mathbf{W}}L_2 \right\rangle_F=0,
\end{equation}
which indicates no gradient interference on this layer.

However, in a convolutional layer, the weight tensor is shared across spatial locations.
Let $\mathbf{W}\in\mathbb{R}^{C_{\text{out}}\times C_{\text{in}}\times K\times K}$ denote the convolutional kernels.
For a fixed output channel $c$, define the corresponding kernel as
$\mathbf{W}_c \triangleq \mathbf{W}_{c,:,:,:}\in\mathbb{R}^{C_{\text{in}}\times K\times K}$.
The gradient with respect to $\mathbf{W}_c$ aggregates contributions over all spatial positions:
\begin{equation}
	\nabla_{\mathbf{W}_c}L
	=
	\sum_{u,v} \boldsymbol{\delta}_{c}[u,v]\ \mathbf{X}[u,v]^\top,
\end{equation}
where $\boldsymbol{\delta}_{c}[u,v]$ is the backpropagated error at location $(u,v)$
and $\mathbf{X}[u,v]$ is the corresponding local input patch.
Therefore, even if two environments activate disjoint spatial regions of a feature map,
their gradients can still update the same kernel parameters $\mathbf{W}_c$ due to weight sharing.

To maintain the no-interference behavior, we apply the hypernet gate at the output-channel level, as illustrated in Fig.~\ref{conv_gate}.
Specifically, the hypernet outputs $\boldsymbol{g}\in\{0,1\}^{C_{\text{out}}}$ and modulates the feature maps as
\begin{equation}
	\mathbf{Y}_{c,:,:}= g_c\, f(\mathbf{Z}_{c,:,:}),
	\qquad c=1,\ldots,C_{\text{out}},
\end{equation}
where $\mathbf{Z}_{c,:,:}$ denotes the $c$-th channel of the convolutional output $\mathbf{Z} \triangleq \text{Conv}(\mathbf{W},\mathbf{X})$. If $g_c=0$, then $\boldsymbol{\delta}_{c}[u,v]=0$ for all $(u,v)$ and hence $\nabla_{\mathbf{W}_c}L=\mathbf{0}$. Therefore, $\left\langle\boldsymbol{g}_1,\boldsymbol{g}_2\right\rangle=0$ will lead to $	\left\langle \nabla_{\mathbf{W}}L_1,\ \nabla_{\mathbf{W}}L_2 \right\rangle_F=0$, indicating that disjoint channel-wise gates yield no interference.

\subsubsection*{Spiking rate consolidation for convolutional layers}
In FC layers, the proposed SRC imposes a regularization term on synaptic weights, where the importance matrix is constructed from the pre- and post-synaptic firing rates. 
Let $\mathbf{u}^F\in\mathbb{R}^{N_{\text{pre}}}$ and $\mathbf{v}^F\in\mathbb{R}^{N_{\text{post}}}$ denote the spiking-rate vectors of the pre- and post-synaptic layers, respectively. We define
\begin{equation}
	\mathbf{\Omega}^F \triangleq \mathbf{v}^F (\mathbf{u}^F)^\top \in \mathbb{R}^{N_{\text{post}}\times N_{\text{pre}}},
\end{equation}
so that the importance of synapse $w_{j,i}^F$ (from pre neuron $i$ to post neuron $j$) is $\Omega^F_{j,i}=v_j^F u_i^F$.
This design follows the Hebbian principle that synapses with stronger pre--post co-activity should be more strongly preserved.

We extend the same idea to convolutional layers.
Consider a convolutional layer with $C_{\text{in}}$ input channels, $C_{\text{out}}$ output channels, kernel size $K\times K$, and stride of $1$ for simplicity. 
The convolutional weight tensor is $\mathbf{W}^C\in\mathbb{R}^{C_{\text{out}}\times C_{\text{in}}\times K\times K}$, and $\mathbf{W}^C_{j,i}\in\mathbb{R}^{K\times K}$ denotes the kernel connecting the $i$-th input channel to the $j$-th output channel. 
Since $\mathbf{W}^C_{j,i}$ is shared across all spatial locations, its importance should depend on the activity levels of the corresponding input and output \emph{channels}. 
Let $\mathbf{u}^C\in\mathbb{R}^{C_{\text{in}}}$ and $\mathbf{v}^C\in\mathbb{R}^{C_{\text{out}}}$ denote the channel-wise spiking-rate summaries (averaged over time and spatial positions) for the input and output channels, respectively. 
We assign the importance of the entire kernel $\mathbf{W}^C_{j,i}$ as $v_j^C u_i^C$, and broadcast it across the spatial kernel entries:
\begin{equation}
	\mathbf{\Omega}^C \triangleq \mathbf{v}^C \otimes \mathbf{u}^C \otimes \mathbf{J}_{K \times K}
	\in \mathbb{R}^{C_{\text{out}}\times C_{\text{in}}\times K\times K},
\end{equation}
where $\mathbf{J}_{K \times K}$ is the $K\times K$ all-ones matrix and $\otimes$ denotes the outer product with broadcasting across the spatial dimensions. 
This construction regularizes kernels connecting highly active input channels to highly active output channels more strongly, preserving task-relevant knowledge with low storage overhead.

\newpage	
\subsection*{Detailed Computation of Energy Consumption}
\subsubsection*{SNNs}
For SNNs, only neurons receiving nonzero spikes perform computations. Hence, the
effective number of operations (NOP) depends on the average firing rate $p\in[0,1]$. 
For a fully connected layer with input dimension $I$ and output dimension $O$,
the NOP is
\begin{equation}
	\mathrm{NOP}_{\mathrm{FC,SNN}} = I O p.
\end{equation}
For a convolutional layer with $C_{\rm in}$ input channels, $C_{\rm out}$ output
channels, kernel size $K\times K$, and output feature map size
$H_{\rm out}\times W_{\rm out}$, we have
\begin{equation}
	\mathrm{NOP}_{\mathrm{Conv,SNN}} = H_{\rm out} W_{\rm out} C_{\rm out} C_{\rm in} K^{2} p.
\end{equation}
Because spike inputs are binary, most operations reduce to accumulations rather
than multiplications. Accordingly, the total energy consumption of an SNN can be
expressed as
\begin{equation}
	E_{\mathrm{SNN}} = 
	\sum_{l\in\mathcal{S}} E_{\rm ac} \,\mathrm{NOP}_{l}\,T +
	\sum_{l\in\mathcal{S}'} E_{\rm mac} \,\mathrm{NOP}_{l}\,T,
\end{equation}
where $\mathcal{S}$ represents the set of spike-based layers, $\mathcal{S}'$
denotes layers operating on analog inputs (e.g., the first layer of the channel
decoder), and $T$ is the total number of simulation time steps.

\subsubsection*{ANNs}
In conventional ANNs, all neurons are active ($p=1$). Thus, the NOP for fully
connected and convolutional layers is given by
\begin{equation}
	\mathrm{NOP}_{\mathrm{FC,ANN}} = I O,
\end{equation}
and
\begin{equation}
	\mathrm{NOP}_{\mathrm{Conv,ANN}} = H_{\rm out} W_{\rm out} C_{\rm out} C_{\rm in} K^{2},
\end{equation}
respectively. Since the energy for matrix multiplication is dominant, we neglect
the energy consumption of the AC operations for bias terms and nonlinear
functions in ANNs and recurrent networks (RNN, LSTM, and Transformer).
Therefore, the total energy consumption of an ANN is given by
\begin{equation}
	E_{\mathrm{ANN}} = E_{\rm mac} \sum_{l} \mathrm{NOP}_{l}.
\end{equation}

\subsubsection*{Sequence models (RNN, LSTM, and Transformer)}
Sequence models extend ANNs and can sequentially handle streaming data over $T$ time steps.
Regarding their energy consumption, an RNN layer with input size $I$ and
hidden size $H$ consists of two matrix multiplications, and the NOP is
\begin{equation}
	\mathrm{NOP}_{\mathrm{RNN}} = T ( I H + H^{2} ).
\end{equation}
LSTM networks contain four gates, leading to
\begin{equation}
	\mathrm{NOP}_{\mathrm{LSTM}} = 4 T ( I H + H^{2} )
\end{equation}
for each layer.
A Transformer block with hidden width $H$ consists of
a self-attention layer followed by two feed-forward FC layers, with the hidden
dimension typically expanded to $4H$. Its NOP can be computed as
\begin{equation}
	\mathrm{NOP}_{\mathrm{Trans}} = T ( 12 H^{2} + 2 T H ).
\end{equation}
Similar to ANNs, all sequence models rely on floating-point MAC operations.
Thus, the total energy consumption is
\begin{equation}
	E_{\mathrm{Seq}} = E_{\mathrm{mac}} \sum_{l} \mathrm{NOP}_{l}.
\end{equation}

\subsubsection*{LDPC}
An $(n,k)$ LDPC code is defined by a generator matrix $\mathbf{G}\in\{0,1\}^{k\times n}$ and a parity-check matrix $\mathbf{H}\in\{0,1\}^{m\times n}$, where $m=n-k$. 
Given an information vector $\mathbf{u}\in\{0,1\}^{1\times k}$, the codeword is generated as
\begin{equation}
	\mathbf{c} = \mathbf{u}\mathbf{G} \bmod 2.
\end{equation}
Let $d_g$ be the average number of nonzero elements per row of $\mathbf{G}$. The encoding energy is modeled as
\begin{equation}
	E_{\mathrm{enc}} = E_{\mathrm{ac}}\, N_{\mathrm{xor}}, 
\end{equation}
where $N_{\mathrm{xor}} \approx n (d_{g}-1)$.

For decoding, the sum–product algorithm (SPA) is employed on the Tanner graph defined by $\mathbf{H}$. Each check node $j$ updates messages according to the $\phi$–sum rule:
\begin{equation}
	L_{j\to i} = \prod_{i'\in\mathcal{N}(j)\setminus i}\!\mathrm{sign}(L_{i'\to j})
	\cdot \phi^{-1}\!\left(\sum_{i'\in\mathcal{N}(j)\setminus i}\phi(|L_{i'\to j}|)\right),
\end{equation}
where $\mathcal{N}(j)$ denotes variable nodes connected to check node $j$ and $\phi(\cdot)$ is defined as
\begin{equation}
	\phi(x) = -\ln\tanh\!\left(\frac{x}{2}\right).
\end{equation}
In this work, we assume a look up table (LUT) method for computing $\phi(\cdot)$ and $\phi^{-1}(\cdot)$ (\textit{70}).%\cite{ldpclut}. 
Each operation requires one memory access and consumes about $E_{mem} = 10 pJ$ (\textit{64}). %\cite{energy_model}. 
Let $E$ denote the number of edges in the Tanner graph and $N_{\text{iter}}$ the iteration number. Then, the decoding energy is approximately estimated as
\begin{equation}
	E_{\mathrm{dec}} \approx N_{\text{iter}}
	\big(E_{\rm mem}\,2E + E_{\rm ac}\,4E\big),
\end{equation}
where the energy for additional operation such as $\mathrm{abs}$ and $\mathrm{sign}$ are neglected.

\subsubsection*{Matrix Operations}
The beamforming experiment involves a lot of complex matrix operations,
such as matrix multiplications and matrix inversions. 
For $\mathbf{A}\in\mathbb{C}^{m\times n}$ and $\mathbf{B}\in\mathbb{C}^{n\times r}$,
the energy for matrix multiplication is given by
\begin{equation}
	E_{\mathrm{mm}} = E_{\mathrm{mac}}\,\mathrm{NOP}_{\mathrm{mm}},
\end{equation}
where
\begin{equation}
	\mathrm{NOP}_{\mathrm{mm}} = 4mnr,
\end{equation}
and the factor~4 reflects the real multiplications required for complex
numbers.

The inversion of a real square matrix $\mathbf{C}\in\mathbb{R}^{n\times n}$
requires approximately $\tfrac{1}{2}n^{3}$ operations (\textit{69}).%\cite{matrix_inverse}. 
For a complex matrix, the required operations are roughly four times larger,
leading to
\begin{equation}
	E_{\mathrm{inv}} \approx 2E_{\mathrm{mac}}\,n^{3}.
\end{equation}

The Cholesky decomposition of a real matrix 
$\mathbf{C}\in\mathbb{R}^{n\times n}$ requires approximately 
$\tfrac{1}{6}n^{3}$ operations (\textit{68}). 
Thus, the energy for the complex case becomes
\begin{equation}
	E_{\mathrm{chol}} \approx \frac{2}{3}E_{\mathrm{mac}}\,n^{3}.
\end{equation}

\subsubsection*{Interpolation}
In OFDM channel estimation, LS recovers non-pilot channel coefficients by
interpolation. We adopt the radial basis function (RBF) interpolation
method (\textit{62}). Specifically, for a time–frequency location $x=(n,k)$ in a
grid of size $N\times K$, the RBF method gives
\begin{equation}
	\hat H(x)
	= \sum_{p=1}^{P} w_p\, \phi\!\left(r^2(x,x_p)\right),
	\qquad 
	\phi(r^2)=e^{-\gamma r^2},
\end{equation}
where $P$ denotes the number of pilots and
$r^2(x,x_p)=\alpha(n-n_p)^2+\beta(k-k_p)^2$ is the scaled squared distance from
$x$ to the $p$-th pilot. The kernel weights $w=[w_1,\ldots,w_P]^T$ are obtained
from the standard RBF system
\begin{equation}
	\Phi w = h,
\end{equation}
where $\Phi\in\mathbb{R}^{P\times P}$ is the pilot kernel matrix with
$\Phi_{pq}=\phi(r^2(x_p,x_q))$ and $h$ contains the LS channel estimates on the
pilot positions. Since the pilot pattern is fixed in the considered OFDM
system, $\Phi$ and its inverse can be precomputed offline. Thus, computing $w$
reduces to a single matrix–vector multiplication,
\begin{equation}
	E_{w} \approx P^{2} E_{\mathrm{mac}}.
\end{equation}

In the interpolation stage, each non-pilot location requires $P$
multiplications and $P$ additions with the precomputed weights, leading to
\begin{equation}
	E_{\mathrm{int}}
	\approx (NK-P)P\,E_{\mathrm{mac}}.
\end{equation}
For complex-valued channels, the real and imaginary parts are interpolated
separately, giving the total energy
\begin{equation}
	E_{\mathrm{tot}}
	\approx 2\,NKPE_{\mathrm{mac}}.
\end{equation}

\newpage
\begin{figure}[h] % Do not use \begin{figure*}
		\centering
		\includegraphics[width=0.6\textwidth]{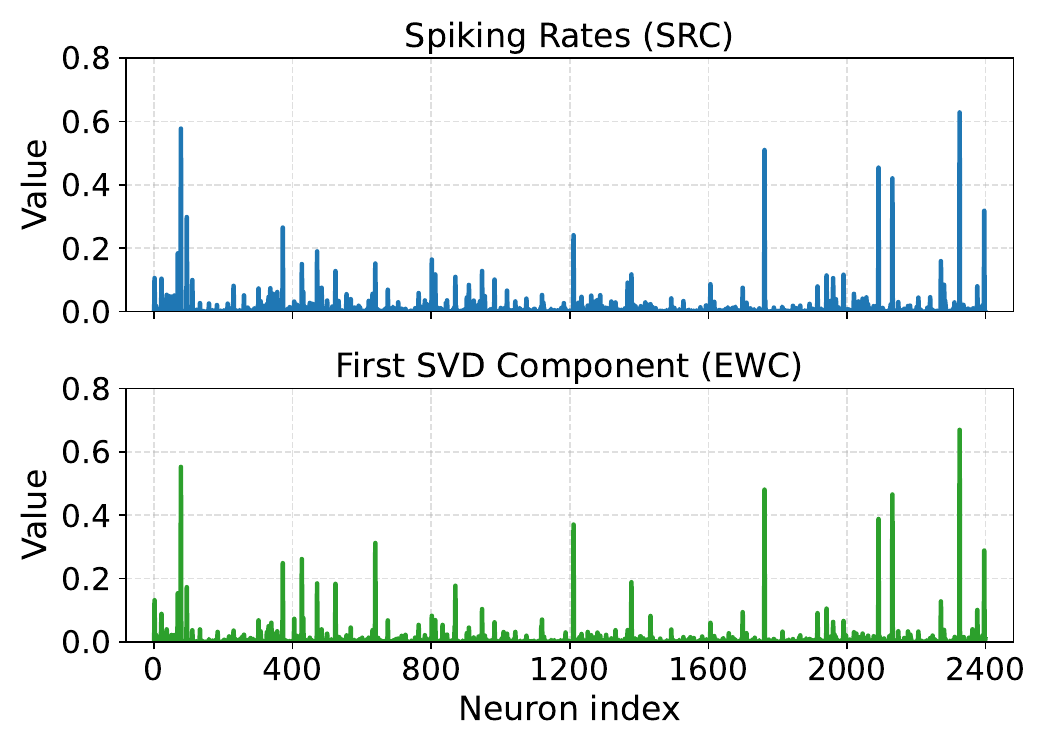} % for an image file named example_figure.*
		% Pick an appriopriate width for the size of the image
		
		% Captions go below figures
		\caption{\linespread{1.0}\selectfont \textbf{Correspondence between spiking activity and Fisher information geometry.} \textbf{Top}: The spiking rate vector across 2400 neurons in the neuromorphic semantic communication task. \textbf{Bottom}: The principal eigenvector of the Fisher information matrix (FIM), indicating the direction of maximum curvature in the parameter space. The alignment between the two panels demonstrates that the computationally efficient spiking rate statistics capture the parameter importance information that is typically estimated from gradient-based FIM.}
		\label{fig:SRMvsFIM} % give each figure a logical label name
	\end{figure}
	
	\newpage
	\begin{figure}[h] % Do not use \begin{figure*}
			\centering
			\includegraphics[width=0.6\textwidth]{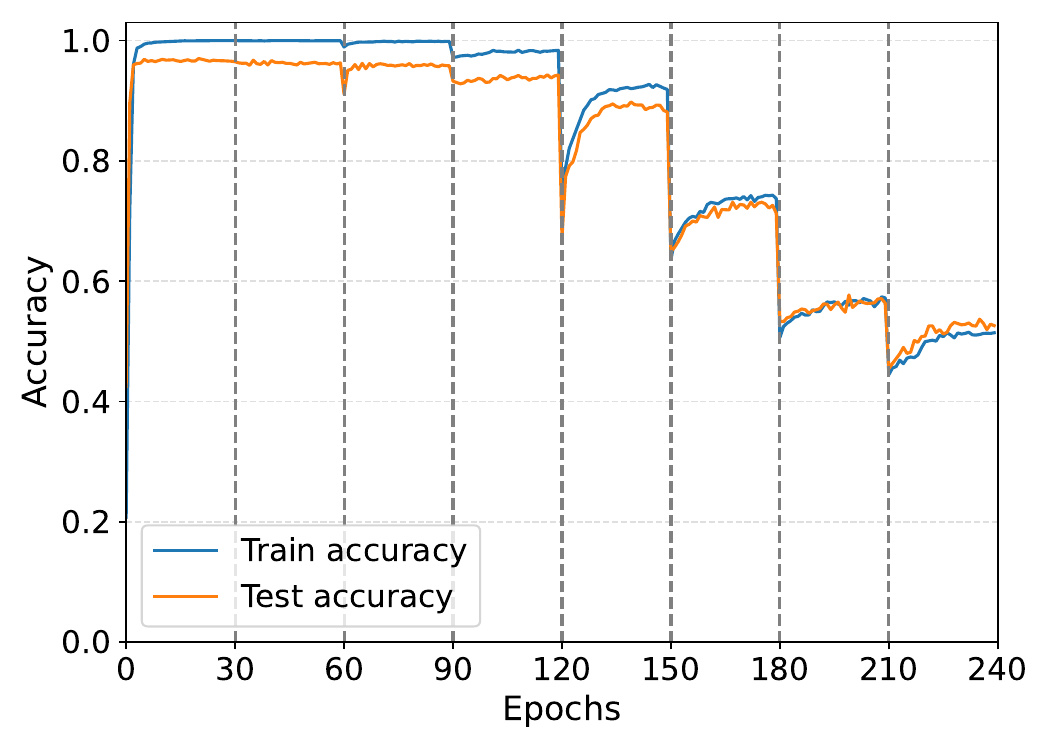} % for an image file named example_figure.*
			% Pick an appriopriate width for the size of the image
			
			% Captions go below figures
			\caption{\linespread{1.0}\selectfont \textbf{Convergence curve of hypernet-based context modulation under dynamically changing environments.} The curves report the training and testing accuracy achieved on the neuromorphic semantic communication task when using the hypernet-generated gate to modulate the backbone network. The model undergoes sequential adaptation across environments whose SNR decreases from $8$~dB to $-20$~dB in uniform steps of $4$~dB; each environment is trained for 30 epochs. As the noise perturbation becomes more severe in later environments, the achievable accuracy decreases accordingly.}
			
			\label{fig:hypernet_gate} % give each figure a logical label name
		\end{figure}

		\newpage
		\begin{figure}[h] % Do not use \begin{figure*}
				\centering
				\includegraphics[width=0.6\textwidth]{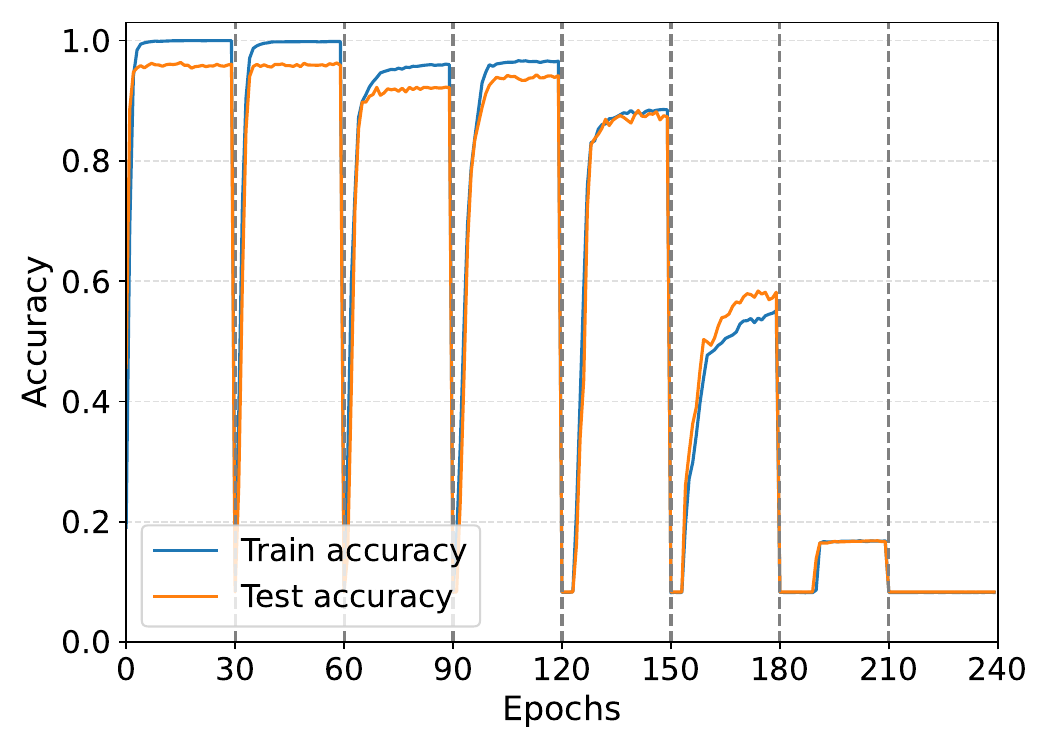} % for an image file named example_figure.*
				% Pick an appriopriate width for the size of the image
				
				% Captions go below figures
				\caption{\linespread{1.0}\selectfont \textbf{Convergence curve of orthogonal gate modulation under dynamically changing environments.} The curves report the training and testing accuracy on the neuromorphic semantic communication task when the backbone network is modulated by mutually orthogonal (non-overlapping) gates, under the same sequential SNR schedule and training protocol as in Fig.~\ref{fig:hypernet_gate}. Orthogonal gating enforces strict parameter isolation across environments, which prevents cross-environment interference but also limits knowledge transfer. Consequently, adaptation to each new environment is slower and may fail to converge in extremely low-SNR scenarios.}

				\label{fig:orthogonal_gate} % give each figure a logical label name
			\end{figure}

			\newpage
			\begin{figure}[h]
				\centering
				\begin{subfigure}[b]{0.49\textwidth}
					\centering
					\includegraphics[width=\textwidth]{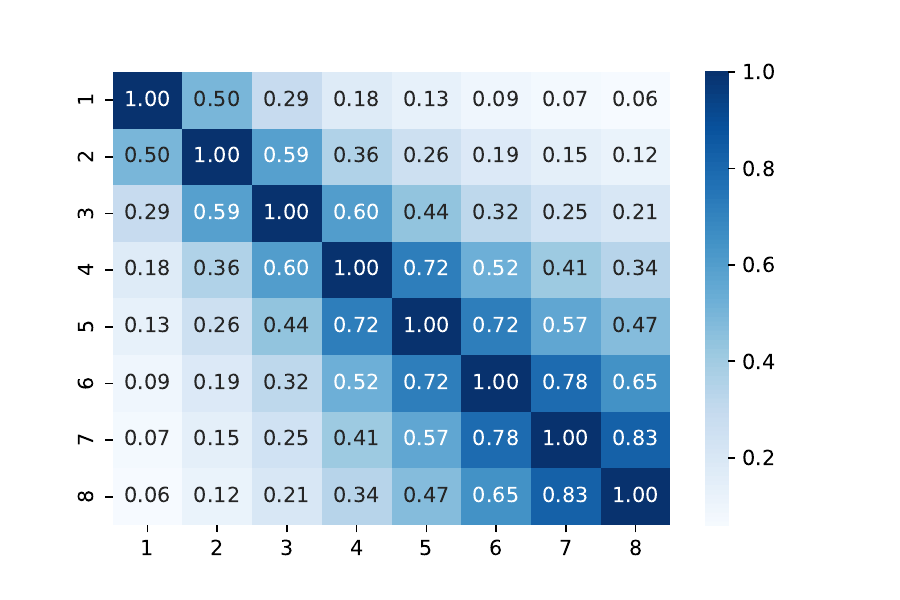}
					\caption{}
				\end{subfigure}
				\hfill 
				\begin{subfigure}[b]{0.49\textwidth}
					\centering
					\includegraphics[width=\textwidth]{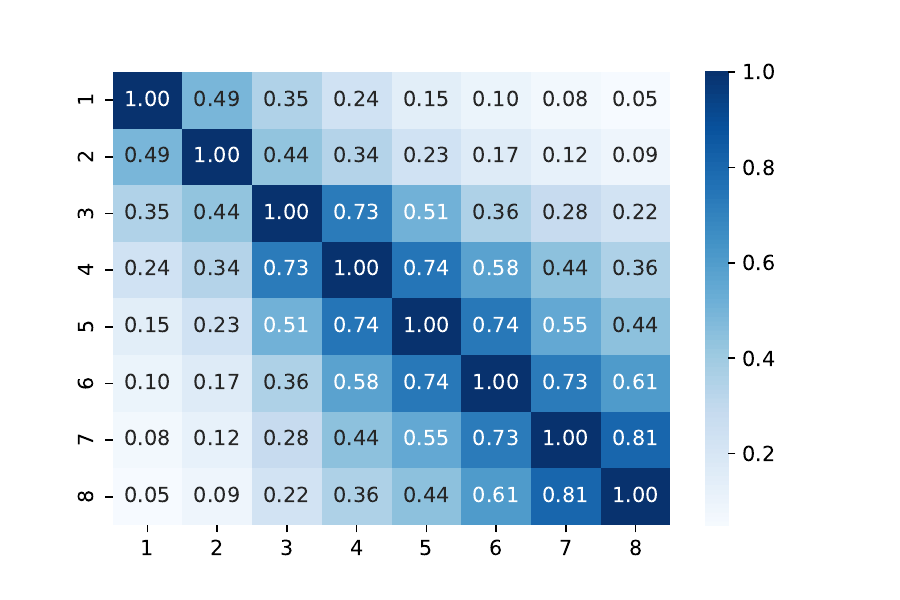}
					\caption{}
				\end{subfigure}
				\caption{\linespread{1.0}\selectfont \textbf{Alignment between model-based EMD and learned hypernet representations in neuromorphic semantic communication task.} 
					(\textbf{a}) The model-based earth mover's distance (EMD) matrix computed between eight unseen wireless environments. Unlike data-driven metrics, this EMD is analytically derived from domain knowledge, allowing for the quantification of physical distribution shifts without requiring specific channel samples. 
					(\textbf{b}) The cosine distance matrix between the hypernet gates generated for the same environments. The hypernet was trained on 80 distinct channel distributions and evaluated on these 8 unseen scenarios.
					The similarity between the physical metric (a) and the learned latent metric (b) indicates that the hypernet captures the underlying physical structure of the wireless channels with limited data access.}
				\label{fig:EMD_vs_gate_semantic}
			\end{figure}
			
			\newpage
			\begin{figure}[h]
				\centering
				\begin{subfigure}[b]{0.49\textwidth}
					\centering
					\includegraphics[width=\textwidth]{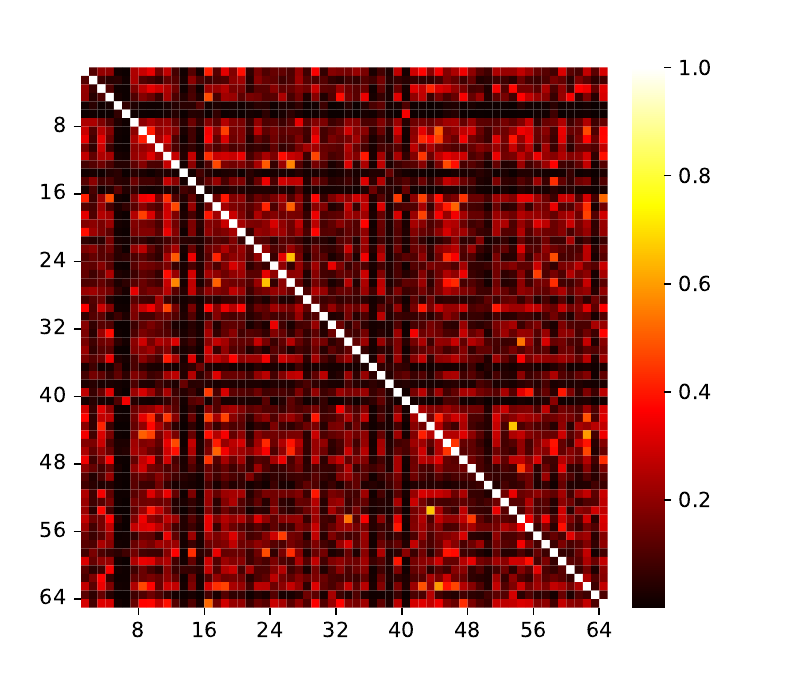}
					\caption{}
				\end{subfigure}
				\hfill % Creates horizontal space between the figures
				% Subfigure 2 (Panel b)
				\begin{subfigure}[b]{0.49\textwidth}
					\centering
					\includegraphics[width=\textwidth]{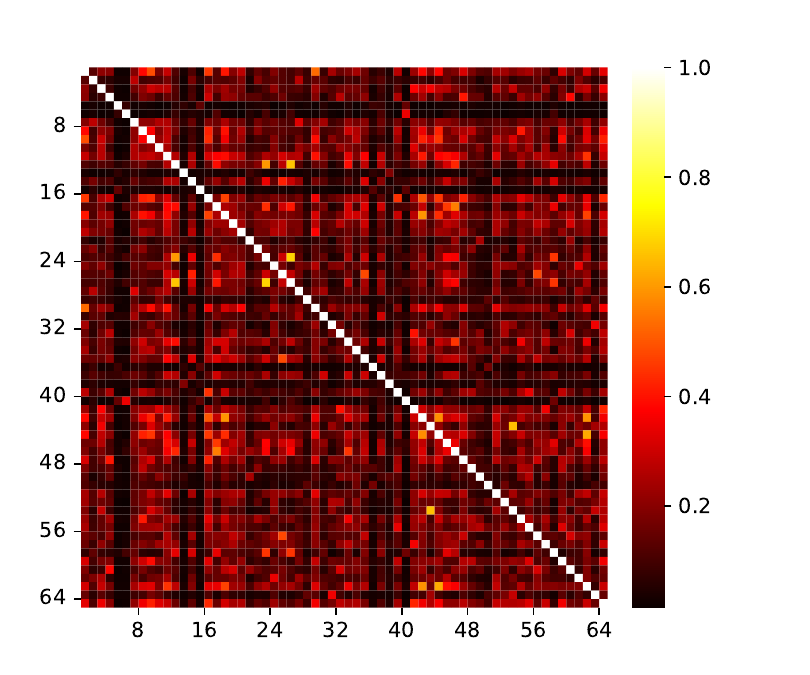}
					\caption{}
				\end{subfigure}
				\caption{\linespread{1.0}\selectfont \textbf{Alignment between data-driven FCD and learned hypernet representations in the MIMO beamforming task.} 
					(\textbf{a}) The data-driven Fréchet channel distance (FCD) matrix computed pairwise among 64 unseen testing environments. Note that for the multi-user setting, we extend FCD by computing the distribution distance over user-wise channel features and applying the Hungarian algorithm to find the minimum EMD over user matchings, yielding a permutation-invariant distance. 
					(\textbf{b}) The cosine distance matrix between the hypernet gates generated for the same environments. 
					We collect 10 pilot samples per user for identifying the environments. The hypernet is trained on 640 distinct environments and evaluated on 64 unseen scenarios. 
					Considering the theoretical space comprises $180^K$ potential environments ($K=6$ users), the training set of 640 covers a very small fraction of the total distribution space. 
					The alignment between (a) and (b) shows that the hypernet captures the structure of channel distributions.}
				\label{fig:EMD_vs_gate_beamforming}
			\end{figure}
			
			\newpage
			\begin{figure}[h]
				\centering
				\includegraphics[width=0.60\textwidth]{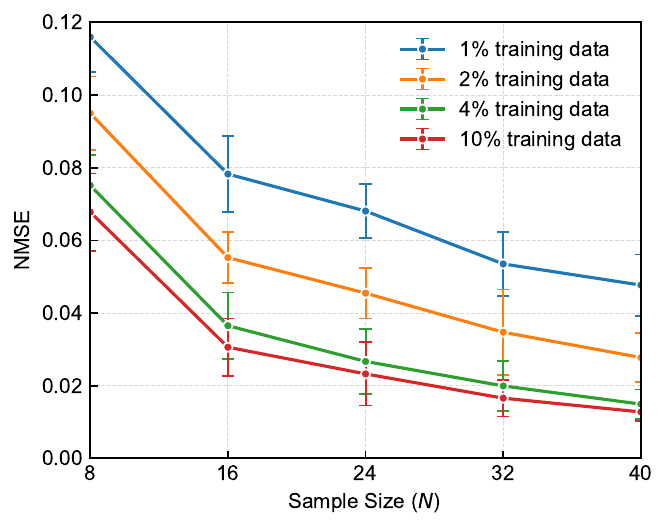}
				\caption{\linespread{1.0}\selectfont \textbf{Assessment of structural alignment between FCD and learned hypernet representations in OFDM channel estimation.} 
					The y-axis represents the NMSE calculated between the data-driven FCD matrix and the learned hypernet cosine distance matrix. 
					The x-axis denotes the number of pilot samples collected per environment. The curves correspond to hypernets trained on 1\%, 2\%, 4\%, and 10\% of the total samples generated from the WINNER II channel model. 
					Given the limited number of distinct scenarios in the WINNER II dataset, this evaluation focuses on the accuracy of environment identification using a few samples rather than generalization between environments. Even with only 1\% of the available training samples, the NMSE drops below 0.1 once the sample size exceeds $N=16$. Moreover, the marginal NMSE gap between the 4\% and 10\% ratios further indicates data efficiency.}
				%	A 4\% training ratio combined with a sample size of 24 achieves an average NMSE lower than 0.02 with reasonably small variance, indicating that the hypernet effectively captures the environmental relationships using a limited training set.}
			\label{fig:EMD_vs_gate_channel_estimation}
		\end{figure}
		
		\newpage
		\begin{figure}[h]
			\centering
			\includegraphics[width=\textwidth]{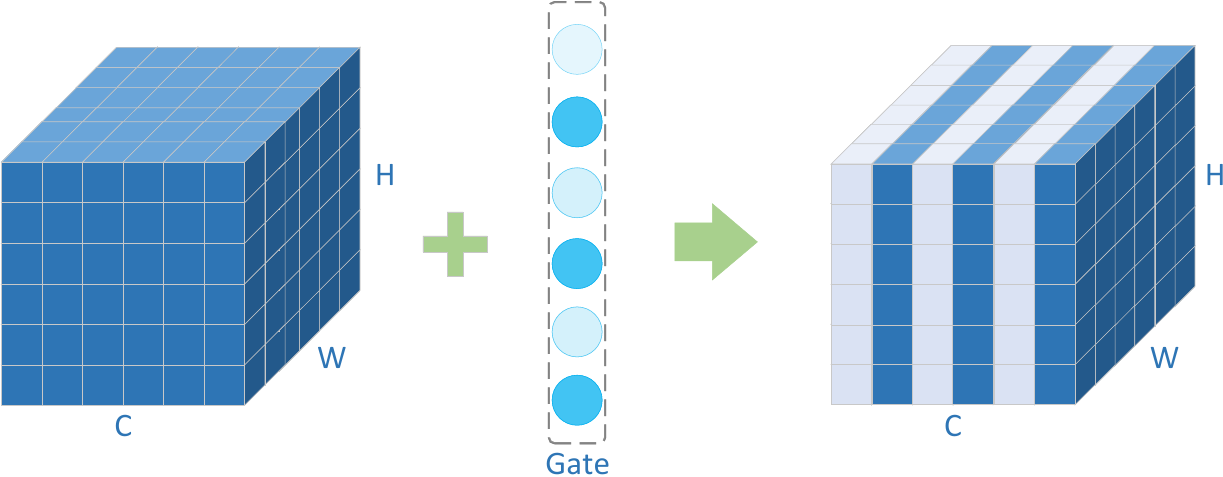}
			\caption{\linespread{1.0}\selectfont \textbf{Hypernet-based context modulator on spiking convolutional layers.} The figure illustrates how the hypernet-generated binary gate modulates the backbone spiking convolutional layers. The output channels are multiplied by the gate. Masked channels are silenced. Those not silenced flow to the next layer.}
			\label{conv_gate}
		\end{figure}
		
		\newpage
		\begin{table}[h]
			\centering
			\caption{\textbf{DeepMIMO Channel Dataset Configuration}}\label{tab:deepmimo}%
			\begin{tabular}{@{}llll@{}}
				\hline
				parameter & value  & parameter & value\\
				\hline
				Scenario    & O1   & Carrier frequency  & 3.5\,GHz  \\
				BS array    & [8 1 8]   & UE array  & [2 1 2]  \\
				Antenna spacing    & 0.5\,wavelength   & Radiation pattern  & Isotropic  \\
				Bandwidth    & 128\,kHz   & Number of paths  & 25  \\
				BS index    & 4   & UE index  & 800--1200  \\
				\hline
			\end{tabular}
		\end{table}
		
		\newpage
		\begin{table}[h]
			\centering
			\caption{\textbf{WINNER II Channel Configuration}}\label{tab:winner}%
			\begin{tabular}{@{}llll@{}}
				\hline
				parameter & value  & parameter & value\\
				\hline
				Center frequency    & 5.25\,GHz   & Bandwidth  & 5\,MHz--20\,MHz  \\
				BS location    & (150\,m, 150\,m, 32\,m)   & UE location  & (110\,m, 120\,m, 1.5\,m)  \\
				UE speed (each axis)    & -0.5\,m/s--0.5\,m/s  & Path loss  & Yes  \\
				Shadow fading    & Yes  & Uniform sampling  & No \\
				\hline
			\end{tabular}
		\end{table}

%%%%%%%%%%%%%%%% END OF MAIN TEXT %%%%%%%%%%%%%%%

\end{document}